\documentclass[pra,aps,twocolumn,superscriptaddress]{revtex4}
\usepackage{color}
\usepackage{bm}
\usepackage{graphicx}
\usepackage{amsbsy}
\usepackage{amsmath}
\usepackage{amsfonts}
\usepackage{amsthm}
\usepackage{listings}
\usepackage{amssymb}

\usepackage{tikz}
\usepackage{tikz-cd}

\usepackage{makecell}
\usepackage{diagbox}

\usepackage{graphicx}
\usepackage{sidecap}

\usepackage{verbatim}

\usepackage{setspace}

\usepackage{enumitem}
\setlist{nolistsep}

\usepackage{quoting}
\quotingsetup{vskip=0pt}

\begin{document}

\theoremstyle{plain}
\newtheorem{theorem}{Theorem}
\newtheorem{lemma}[theorem]{Lemma}
\newtheorem{corollary}[theorem]{Corollary}
\newtheorem{conjecture}[theorem]{Conjecture}
\newtheorem{proposition}[theorem]{Proposition}

\theoremstyle{definition}
\newtheorem{definition}{Definition}

\theoremstyle{remark}
\newtheorem*{remark}{Remark}
\newtheorem{example}{Example}

\def\be{\begin{equation}}
\def\ee{\end{equation}}
\def\ba{\begin{align}}
\def\ea{\end{align}}

\newcommand{\mE}{\mathcal{E}}
\newcommand{\mU}{\mathcal{U}}
\newcommand{\mA}{\mathcal{A}}
\newcommand{\mF}{\mathcal{F}}
\newcommand{\mI}{\mathcal{I}}
\newcommand{\mH}{\mathcal{H}}
\newcommand{\mL}{\mathcal{L}}
\newcommand{\mM}{\mathcal{M}}
\newcommand{\mT}{\mathcal{T}}
\newcommand{\mN}{\mathcal{N}}

\newcommand{\fm}{\mathcal{F}_{\bf{m}}}
\newcommand{\am}{\mathcal{A}^{\textbf{m}}}
\newcommand{\dm}{\mathcal{D}(\mathrm{H}_{{\bf m}})}
\newcommand{\lr}{\rangle\langle}
\newcommand{\la}{\langle}
\newcommand{\ra}{\rangle}
\newcommand{\tr}{{\rm Tr}}

\newcommand{\mc}[1]{\mathcal{#1}}
\newcommand{\mbf}[1]{\mathbf{#1}}
\newcommand{\mbb}[1]{\mathbb{#1}}
\newcommand{\mrm}[1]{\mathrm{#1}}

\newcommand{\bra}[1]{\langle #1|}
\newcommand{\ket}[1]{|#1\rangle}
\newcommand{\braket}[3]{\langle #1|#2|#3\rangle}
\newcommand{\ip}[2]{\langle #1|#2\rangle}
\newcommand{\op}[2]{|#1\rangle \langle #2|}

\newcommand{\mbN}{\mathbb{N}}

\definecolor{eric}{rgb}{0,.5,.2}
\newcommand{\eric}[1]{{\color{eric} #1}}

\newcommand{\review}[1]{{\color{red} #1}}

\title{Are Incoherent Operations Physically Consistent? -- A Critical Examination of Incoherent Operations}   
\author{Eric Chitambar}\email{echitamb@siu.edu}
\affiliation{Department of Physics and Astronomy, Southern Illinois University,
Carbondale, Illinois 62901, USA}
\author{Gilad Gour}\email{gour@ucalgary.ca}
\affiliation{
Department of Mathematics and Statistics,
University of Calgary, AB, Canada T2N 1N4} 
\affiliation{
Institute for Quantum Science and Technology,
University of Calgary, AB, Canada T2N 1N4}

\date{\today}


\begin{abstract}

Considerable work has recently been directed toward developing resource theories of quantum coherence. In most approaches, a state is said to possess quantum coherence if it is not diagonal in some specified basis.  In this letter we establish a criterion of physical consistency for any resource theory in terms of physical implementation of the free operations, and we show that all currently proposed basis-dependent theories of coherence fail to satisfy this criterion.  We further characterize the physically consistent resource theory of coherence and find its operational power to be quite limited.  
After relaxing the condition of physical consistency, we introduce the class of dephasing-covariant incoherent operations, present a number of new coherent monotones based on relative R\'{e}nyi entropies, and study incoherent state transformations under different operational classes.  In particular, we derive necessary and sufficient conditions for qubit state transformations and show these conditions hold for all classes of incoherent operations.
\end{abstract}

\maketitle

Resource theories offer a powerful framework for understanding how certain physical properties naturally change within a physical system.  A general resource theory for a quantum system is characterized by a pair $(\mc{F},\mc{O})$, where $\mc{F}$ is a set of ``free'' states and $\mc{O}$ is a set of ``free'' quantum operations.  Any state that does not belong to $\mc{F}$ is then deemed a resource state.  Entanglement theory provides a prototypical example of a resource theory in which the free states are the separable or unentangled states, and the free operations are local operations and classical communication (LOCC) \cite{Plenio-2007a, Horodecki-2009a}.  Other examples includes the resource theories of athermality~\cite{Janzing-2000a, Brandao-2013a}, asymetry~\cite{Gour-2008a,Marvian2013,Marvian2014}, and non-stabilizer states for quantum computation~\cite{Veitch-2014a}. 

Any pair $(\mc{F},\mc{O})$ defines a resource theory, provided the operations of $\mc{O}$ act invariantly on $\mc{F}$; i.e. $\mc{E}(\rho)\in\mc{F}$ for all $\rho\in\mc{F}$ and all $\mc{E}\in\mc{O}$.  However, this is just a mathematical restriction placed on the maps belonging to $\mc{O}$.  It does not imply that $\mc{E}\in\mc{O}$ can actually be physically implemented without generating or consuming additional resource.  The issue is a bit subtle here since in quantum mechanics, physical operations on one system ultimately arise from unitary dynamics and projective measurements on a larger system, a process mathematically described by a Stinespring dilation \cite{Paulsen-2003a}.  A resource theory $(\mc{F},\mc{O})$ defined on system $A$ is said to be \textit{physically consistent} if every free operation in $\mc{E}\in\mc{O}$ can be obtained by an auxiliary state $\hat{\rho}_B$, a joint unitary $U_{AB}$, and a projective measurement $\{P_k\}_k$ that are all free in an extended resource theory $(\mc{F}',\mc{O}')$ defined a larger system $AB$, for which $\mc{O}=\text{Tr}_B\mc{O}':=\{\text{Tr}_B(\rho_{AB}):\rho_{AB}\in\mc{O}'\}$.


\begin{table}[t]
\linespread{1}\selectfont{\small 
\begin{tabular}{|l|c|c|}\hline
\theadfont\diagbox[width=11em]{Resource}{Operations}&
\thead{Physically\\Consistent}&\thead{Physically\\Inconsistent}\\    \hline
Entanglement & LOCC & SEP, NE \\    \hline
Coherence & PIO & SIO, DIO, IO, MIO \\ \hline
\end{tabular}
\caption{The class of Physically Incoherent Operations (PIO) introduced in this letter represents the coherence analog to LOCC in terms of being a physically consistent resource theory.  The previously studied Strictly Incoherent Operations (SIO), Incoherent Operations (IO) and Maximally Incoherent Operations (MIO) represent relaxations of PIO in the same way that Separable (SEP) and Non-Entangling (NE) operations are relaxations of LOCC.  We further introduce the new class of Dephasing-covariant Incoherent Operations (DIO).}}
\end{table}

Arguably a physically consistent resource theory is more satisfying than an inconsistent one.  Indeed, without physical consistency, the notions of ``free'' and ``resource'' have very little physical meaning since resources must ultimately be consumed to implement certain operations that are supposed to be `` free.''  As an analogy, if a car wash offers to wash your car for free, but only after you go across the street and purchase an oil change from their business partner, is the ``car washing operation'' really free? 

At the same time, physically inconsistent resource theories can still be of interest.  Consider again entanglement.  LOCC renders a physically consistent resource theory of entanglement since any LOCC operation can be implemented using only local unitaries and projections.  However, often one considers more general operational classes such as separable operations (SEP) or the full class of non-entangling operations (NE) \footnote{It is also common to consider the class of positive partial-transpose preserving operations (PPT) as a relaxation of LOCC.  However, the PPT resource theory has a different set of free states than LOCC; namely all PPT entangled state are free in the former while they are not in the latter.}.  The motivation for using SEP is that it possesses a much nicer mathematical structure than LOCC without being too much stronger.  In contrast, one may turn to NE when seeking maximal strength among all operations that cannot generate entanglement.  Nevertheless, despite being appealing objects of study, both SEP and NE represent physically inconsistent resource theories of entanglement.

In this letter, we analyze some of the recently proposed resource theories of quantum coherence \cite{Aberg-2004a, Baumgratz-2014a, Yadin-2015a, Streltsov-2015a}.  We observe that none of these offer a physically consistent resource theory, and the true analog to LOCC in coherence theory has been lacking.  We identify this hitherto missing piece as the class of \textit{physically incoherent operations} (PIO), and we provide its characterization.  The operations previously used to study coherence are much closer akin to SEP and NE in entanglement theory, and we clarify what sort of physical interpretations can be given to these operations.  

While we find that PIO allows for optimal distillation of maximal coherence from partially coherent pure states in the asymptotic limit of many copies, the process is strongly irreversible.  That is, maximally coherent states cannot be diluted into weakly coherent states at a nonzero rate, and they are thus curiously found to be the \textit{least} powerful among all coherent states in terms of asymptotic convertibility.  Given this limitation of PIO and its similar weakness on the finite-copy level, it is therefore desirable from a theoretical perspective to consider more general operations. 
Consequently, we shift our focus to the development of coherence resource theories under different relaxations of PIO.  To this end, our main contributions are as follows.

We introduce the class of dephasing-commuting incoherent operations (DIO), which to our knowledge has never discussed before in literature.  We provide physical motivation for DIO and show that these operations are just as powerful as Maximal Incoherent Operations (MIO) when acting on qubits.
We then study the class MIO and show, somewhat surprisingly, that MIO can increase the Schmidt rank of pure states. New coherence measures based on the relative R\'{e}nyi entropies are presented for DIO and MIO.  Finally, we study the resource theory of $N$-asymmetry,
where $N$ is the group of all diagonal unitaries with respect to the incoherent basis. We show that for physical systems without $U(1)$-translation symmetry, the resource theory of $N$-asymmetry can characterize coherence more adequately than the resource theory of $U(1)$-asymmetry.

Quantum coherence has traditionally referred to the presence of off-diagonal terms in the density matrix.  For a given (finite-dimensional) system, a complete basis $\{\ket{i}\}_{i=1}^d$ for the system is specified, accounting for all degrees of freedom, and a state is said to lack coherence (or be ``incoherent'') with respect to this basis if and only if its density matrix is diagonal in this basis \cite{Blum-2012a, Cohen-Tannoudji-1992a}.  We will refer to this as a \textit{basis-dependent definition of coherence}, and accordingly, a basis-dependent resource theory of coherence identifies the free (or ``incoherent'') states $\mc{I}$ as precisely the set of diagonal density matrices in the fixed incoherent basis \footnote{One can also adopt a more general notion of coherence based on asymmetry \cite{RobIman}.  In this setting, coherence is not identified with the off-diagonal elements of a density matrix, and consequently, one have states $\ket{0}$, $\ket{1}$, $\ket{\psi}=\sqrt{1/2}(\ket{0}+\ket{1})$ and $\rho=1/2(\op{0}{0}+\op{1}{1})$ that are all considered to be incoherent.  This is a departure from traditional parlance in which $\ket{\psi}$ is called a \textit{coherent} superposition whereas $\rho$ is an \textit{incoherent} superposition (see Appendix for more discussion).}.

When it comes to identifying the free (or ``incoherent'') operations, different proposals have been made.  We focus on the following three operational classes.  A CPTP map $\mc{E}$ is said to be: a Maximal Incoherent Operation (MIO) if $\mc{E}(\rho)\in \mc{I}$ for every $\rho\in\mc{I}$ \cite{Aberg-2004a, Aberg-2006a}; an Incoherent Operation (IO) if $\mc{E}$ has a Kraus operator representation $\{K_n\}_{n}$ such that $K_n\rho K_n^\dagger/\tr[K_n\rho K_n^\dagger]\in\mc{I}$ for all $n$ and $\rho\in\mc{I}$ \cite{Baumgratz-2014a}; a Strictly Incoherent Operation (SIO) if $\mc{E}$ has a Kraus operator representation $\{K_n\}_{n}$ such that $K_n\Delta(\rho) K_n^\dagger=\Delta(K_n\rho K_n^\dagger)$ for all $n$ \cite{Yadin-2015a}, where $\Delta$ is the completely dephasing map $\Delta:\rho\mapsto\sum_{i=1}^{d_A}\op{i}{i}\rho\op{i}{i}$. 

In each of these approaches, the allowed unitary operations and projective measurements are the same.  The set of all incoherent unitary matrices forms a group which we denote by $G$.  For a $d$-dimensional system, the group $G$ consists of all $d\times d$ unitaries of the form $\pi u$, where $\pi$ is a permutation matrix and $u$ is a diagonal unitary matrix (with phases on the diagonal).  We denote by $N\cong U(1)^d$ the group of diagonal unitary matrices and by $\Pi$ the group of permutation matrices. Note that $N$ is a normal subgroup of $G$, and $G=N\rtimes \Pi$ is the semi-direct product of $N$ and $\Pi$.  Likewise, an incoherent projective measurement consists of any complete set of orthogonal projectors $\{P_j\}$ with each $P_j$ being diagonal in the incoherent basis. 

It is crucial that a physical resource theory possess a well-defined extension to multiple systems if one allows for generalized measurements, simply because the latter describes a process that is carried out on more than one system.  
A natural requirement for any physical resource theory of coherence is that it satisfies the \textit{no superactivation postulate}; that is, if $\rho$ and $\sigma$ lack quantum coherence, then so must the joint state $\rho\otimes\sigma$.  Combining the basis-dependent definition of coherence with the no superactivation postulate immediately fixes the structure of multipartite incoherent states.  If $\{\ket{i}_A\}_{i=1}^{d_A}$ and $\{\ket{j}_B\}_{j=1}^{d_B}$ are defined to be the incoherent bases for systems $A$ and $B$ respectively, then the superactivation postulate forces $\{\ket{i}_A\ket{j}_{B}\}_{i,j=1}^{d_A,d_B}$ to be the incoherent basis for the joint system $AB$.

\begin{figure}[t]
\includegraphics[scale=.25]{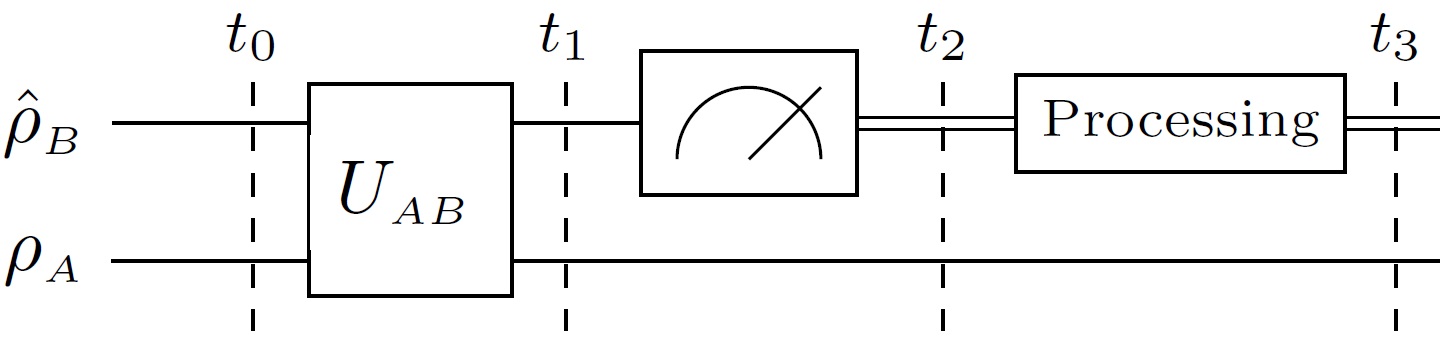}
\caption{\label{Fig:Measure}  This figure depicts the general process of implementing an incoherent operation on the joint system $AB$ whose reduced action on $A$ is the incoherent CPTP map $\rho_A\mapsto\hat{\mc{E}}(\rho_A)$.  A second system $B$ is introduced in an incoherent state $\hat{\rho}^B$.  Both the unitary $U_{AB}$ and projective measurement are coherence non-generating.  All measurement outcomes are stored in a classical register of system $B$ so that the joint system is in a QC state at time $t_2$.  Only maps $\hat{\mc{E}}$ implemented in this way are physically consistent within a resource-theoretic picture.}
\end{figure}

The fact that the incoherent basis takes tensor product form when considering multiple systems has strong consequences for the physical consistency of incoherent operations.  Every physical operation on some system, say $A$, can be decomposed into a three-step process as depicted in Fig. \ref{Fig:Measure}.  If this operation is free within a physically consistent framework, then (i) a joint incoherent unitary $U_{AB}$ is applied immediately prior to time $t_1$ on the input state $\rho_A$ and some fixed incoherent state $\hat{\rho}_B$, (ii) an incoherent projective measurement is applied immediately prior to time $t_2$ with system $B$ encoding the measurement outcome as a classical index, and (iii) a classical processing channel is applied to the measurement outcomes immediately prior to $t_3$.  Note that at time $t_2$, the joint state is a quantum-classical (QC) state $\omega_{AB}=\sum_{j=1}^t\rho_{A,j}\otimes\op{j}{j}_B$, where 
\[\rho_{A,j}=\tr_B[(\mbb{I}_A\otimes P_j) U_{AB}(\rho_A\otimes\hat{\rho}_B)U_{AB}^\dagger ].\]  With the classical processing, the final state of system $A$ at time $t_3$ is given by $\mc{E}(\rho_A):=\sum_{k=1}^{t'}\rho'_{A,k}\otimes\op{k}{k}$, where $\rho'_{A,k}=\sum_{j=1}^tp_{k|j}\rho_{A,j}$ for some channel $p_{k|j}$.  We define the class of physical incoherent operations (PIO) to be the set of all CPTP maps $\hat{\mc{E}}$ that can be obtained in this way.  
\noindent The following characterization of PIO is derived in the Appendix.
\begin{proposition}
A CPTP map $\hat{\mc{E}}$ is a physically incoherent operation if and only if it can be expressed as a convex combination of maps each having Kraus operators $\{K_j\}_{j=1}^r$ of the form
\begin{equation}
K_j=U_jP_j=\sum_{x}e^{i\theta_x}\op{\pi_j(x)}{x}P_{j},
\end{equation}
where the $P_j$ form an orthogonal and complete set of incoherent projectors on system $A$ and $\pi_j$ are permutations.
\end{proposition}

From the proposition above it is easy to see that PIO $\subset$ SIO $\subset$ IO $\subset$ MIO, with PIO being a strict subset of the other three.  To understand the physical differences between these operations let us return to Fig. \ref{Fig:Measure} and for the sake of the following discussion, assume that the measurement between times $t_1$ and $t_2$ is a rank-one projection into the incoherent basis $\{\ket{j}\}_{j=1}^{d_B}$.  Then the joint state at time $t_2$ takes the form $\sum_{j=1}^{d_B}K_j\rho_A K_j^\dagger\otimes\op{j}{j}_B$ for Kraus operators $\{K_j\}_{j=1}^{d_B}$.  Suppose now that the input $\hat{\rho}_A$ is incoherent so that initial joint state $\hat{\rho}_A\otimes\hat{\rho}_B$ is also incoherent.  If the final state at time $t_3$ is always incoherent, regardless of the coherence generated during the intermediate times, then the operation is a maximally incoherent operation (MIO).  If the QC joint state at time $t_2$ is always incoherent, then the operation is an incoherent operation (IO).  If the joint state at time $t_1$ is always incoherent, then the operation is a physically incoherent operation (PIO), provided the subsequent projective measurement is incoherent.  Conversely, every IO/MIO operation can be implemented using the scheme of Fig. \ref{Fig:Measure} by taking the size of system $B$ to be sufficiently large.  Where do SIO operations fit in this picture?  Despite the discussion presented in Ref. \cite{Yadin-2015a}, it is not entirely clear \footnote{In Ref. \cite{Yadin-2015a} it is claimed that if a set of Kraus operators $\{K_n\}_n$ represents a strictly incoherent operation - so that $K_n\Delta(\rho) K_n^\dagger=\Delta( K_n\rho K_n^\dagger)$ - then the QC state at time $t_2$ can always be obtained by a unitary $U_{AB}$ of the form $
U_{AB}:\ket{i}_A\ket{0}_B \to\ket{\pi(i)}_A\ket{\psi_i}_B$, where $\pi$ is a permutation and $\ket{\psi_i}$ is an arbitrary state.  However in general this is not true.  As we show in the appendix, the correct form of $U_{AB}$ for an SIO has the form $
U_{AB}=\sum_{i,k} c_{k i}\op{\pi_k(i)}{i}\otimes\op{k}{0}$, for different permutations $\pi_k$.}.

The class PIO is a rather restricted class of operations.  For instance, suppose that $\ket{\psi}$ and $\ket{\phi}$ are any two pure states with rank$[\Delta(\psi)]=\text{rank}[\Delta(\phi)]$.  Then $\ket{\psi}$ can be converted to another $\ket{\phi}$ using PIO if and only if the states are unitarily equivalent.  

The power of PIO is improved somewhat on the many-copy level.  One can easily show that a state $\ket{\psi}$ can be asymptotically converted via PIO into the maximally coherent qubit state $\ket{+}=\sqrt{1/2}(\ket{0}+\ket{1})$ at a rate equaling the von Neumann entropy of the state $\Delta(\op{\psi}{\psi})$, which is optimal (see Ref. \cite{Winter-2015a} for details of the incoherent projective measurement).  On the other hand, the asymptotic conversion rate of $\ket{+}$ into any weakly coherent state $\ket{\psi}$ is strictly zero.  The proof of this fact reveals an interesting relationship between quantum coherence and communication complexity in LOCC.  Observe that for any PIO transformation $\ket{\psi}\to\ket{\varphi}$, there exists a \textit{zero communication} LOCC protocol that transforms $\ket{\psi^{(mc)}}\to\ket{\varphi^{(mc)}}$, where $\ket{\psi^{(mc)}}$ and $\ket{\varphi^{(mc)}}$ are maximally correlated extensions of $\ket{\psi}$ and $\ket{\varphi}$; i.e. $\ket{\psi^{(mc)}}=\sum_{i}\sqrt{p_i}\ket{ii}_{AB}$ when $\ket{\psi}=\sum_i\sqrt{p_i}\ket{i}_A$.  Thus obtaining $nR$ copies of $\ket{\varphi}$ from multiple copies of $\ket{+}$ implies that $nR$ copies of $\ket{\varphi^{(mc)}}$ can be be obtained from multiple EPR pairs using no communication.  However, this contradicts the communication lower bounds \cite{Harrow-2003a, Hayden-2003a} that require nonzero communication to reliably obtain $nR$ copies of $\ket{\varphi^{(mc)}}$ from a source of EPR pairs.  Hence, rather bizarrely, in PIO theory the maximally coherent state is the weakest as it cannot be transformed into other that is not related by an incoherent unitary.

The weakness of PIO means that the constraint of physical consistency is too strong if one wishes to have a less degenerate resource theory of coherence.  This provides motivation to relax the constraint of physical consistency and to consider more general resource theories such as SIO/IO/MIO.  We now turn to one such theory that has not been previously discussed, but in some sense it is the most natural one to consider.

\medskip

\noindent\textit{Dephasing-Covariant Incoherent Operations.}    The family of Dephasing-Covariant Incoherent Operations (DIO) consists of all maps that commute with $\Delta$.  Recall that in general, for a collection of operations $T$, a CPTP map $\mc{E}$ is said to be $T$-covariant if $[\mc{E},\tau]=0$ for all $\tau\in T$.  DIO can be seen as a natural extension of PIO in light of the following theorem, whose proof is given in the Appendix.
\begin{theorem}$\;$\label{gcom}
{\bf{\rm (a)}} Let $G$ be the group of incoherent unitaries.  Then, $[\mU,\Delta]=0$ iff $U\in G$.  {\bf{\rm (b)}} A CPTP map $\mE$ is G-covariant iff
\begin{align}
\label{Eq:G-cov-form}
\mE(\rho)&=q_1\rho+
\frac{q_2}{d-1}\left(I-\Delta(\rho)\right)+
\frac{q_3}{d-1}\left(d\Delta(\rho)-\rho\right)
\end{align}
for some $q_i\geq 0$ with $\sum_{i=1}^{3}q_i=1$.  {\bf{\rm (c)}}  A CPTP map $\mE$ is {\upshape PIO}-covariant iff it has the form of Eq. \eqref{Eq:G-cov-form} with $q_2=0$.
\end{theorem}
\noindent From part (c) of Theorem \ref{gcom}, the commutant of PIO consists of the family of channels $\Delta_\lambda(\rho):=(1-\lambda)\rho+\lambda\Delta(\rho)$ for $\lambda\in[0,1]$.  The class DIO therefore generalizes PIO in that it is largest operational class sharing the same commutant as PIO (see Fig.2).  

Operational covariance is an important physical property as it describes an order invariance in performing a two-step process.  DIO are of particular interest when observing how the probabilities $p_i=\bra{i}\rho\ket{i}$ transform under a map $\mc{E}$.  If $\mc{E}$ is DIO, then an experimenter can put $\rho$ through any channel $\Delta_\lambda$ before applying $\mc{E}$ without changing the probabilities $p_i$.  Note that DIO can also be seen as an extension of SIO to general channels.  We next turn attention to the presentation of various coherence measures.

\begin{SCfigure}[1.3][t]
\centering
\caption{\label{Fig:Commute}A DIO map $\mc{E}$ commutes with every channel $\Delta_\lambda(\rho):=(1-\lambda)\rho+\lambda\Delta(\rho)$ for $\lambda\in[0,1]$.}
\begin{tikzcd}[column sep=large, row sep=normal]
\rho \arrow{r}{\Delta_\lambda} \arrow[swap]{d}{\mc{E}} & \Delta_\lambda(\rho) \arrow{d}{\mc{E}} \\
\mc{E}(\rho)  \arrow{r}{\Delta_\lambda} & \rho'
\end{tikzcd}
\end{SCfigure}
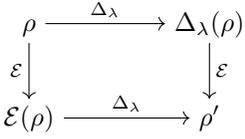

\medskip

\noindent\textit{New Coherence Measures.}  For a given coherence resource theory, the essential property of any coherence measure is that it monotonically decreases under the free operations.  The following introduces coherence monotones for the operational classes DIO and MIO.  Since MIO is the largest class of coherence non-generating operations, the MIO monotones hold for \textit{any} resource theory of coherence (including IO and SIO).  These monotones are based on relative R\'{e}nyi entropies, but they actually arise within a much more general family of monotones described in the Appendix.

The relative R\'{e}nyi entropy $D_\alpha$ and quantum relative R\'{e}nyi entropy $D_\alpha^{(q)}$ of $\rho$ to $\sigma$ are defined as $D_{\alpha}(\rho\|\sigma):=\frac{1}{\alpha-1}\log\tr(\rho^\alpha\sigma^{1-\alpha})$ and $D_{\alpha}^{(q)}(\rho\|\sigma):=\frac{1}{\alpha-1}\log\tr[(\sigma^{\frac{1-\alpha}{2\alpha}}\rho\sigma^{\frac{1-\alpha}{2\alpha}})^\alpha]$.
The function $D_\alpha$ is known to be contractive for $\alpha\in[0,2]$ while $D_\alpha^{(q)}$ is contractive for $\alpha\in[1/2,\infty]$.  Define
\begin{align}
C_\alpha(\rho)&=\min_{\sigma\in\mc{I}}D_\alpha(\rho||\sigma),\;\;\;\; \alpha\in[0,2];\label{Eq:Renyi-DIO-1}\\
C^{(q)}_\alpha(\rho)&=\min_{\sigma\in\mc{I}}D^{(q)}_\alpha(\rho||\sigma),\;\;\;\;\alpha\in[1/2,\infty].\label{Eq:Renyi-DIO-2}
\end{align}
The measures $C_\alpha$ and $C_\alpha^{(q)}$ are monotones under MIO, and they generalize other measures of coherence studied in the literature.  For instance, when taking $C_\alpha$ in the limit $\alpha\to 1$, the \textit{Relative Entropy of Coherence} \cite{Baumgratz-2014a} is obtained: $C_\alpha(\rho)\to C_{rel}(\rho)$.  Likewise, taking $C_{\alpha}^{(q)}$ in the limit $\alpha\to\infty$ yields the $\log[1+ C_R]$, where $C_R$ is the \textit{Robustness of Coherence} \cite{Piani-2016a}:
\begin{equation}
 C_R(\rho)=\min_{t\geq 0}\left\{t\;\Big|\;\frac{\rho+t\sigma}{1+t}\in\mI,\;\sigma\geq 0\right\}.\notag
\end{equation}
Furthermore, for pure states $\ket{\psi}=\sum_i\sqrt{p_i}\ket{i}$ and $\alpha\in[1/2,\infty]$, both $C_{1/\alpha}(\psi)$ and $C^{(q)}_{\alpha/(2\alpha-1)}(\psi)$ reduce to the R\'{e}nyi entropy $S_\alpha(p_i)$ of the distribution $p_i$:
\begin{equation}
\label{Eq:Renyi-Pure}
C_{1/\alpha}(\psi)=\tfrac{1}{1-\alpha}\log\sum_i p_i^\alpha=S_\alpha(p_i),\;\;\;\;\alpha\in[1/2,\infty].
\end{equation}
The following theorem exemplifies the power of MIO for pure state transformations.
\begin{theorem}
\label{Thm:MIO-pure}
Let $|\psi\ra=\sqrt{p_0}|0\ra+\sqrt{p_1}|1\ra$ and $|\psi\ra=\sum_{y=1}^{d'}\sqrt{q_y}|y\ra$, where
$q_y>0$ and $d'>2$. Then, $|\psi\ra$ can be converted to $|\phi\ra$ if and only if $p_0=p_1=1/2$ and $
\sum_{y=1}^{d'}\sqrt{q_y}\leq\sqrt{2}$.
\end{theorem}
\noindent This result is remarkable since it shows that the Schmidt rank of $\Delta(\psi)$ is not a monotone under MIO.  In contrast, the diagonal rank is indeed a coherence monotone for IO \cite{Baumgratz-2014a}.   More generally, this example implies that all R\'{e}nyi entropies of the distribution $p_i=|\ip{\psi}{i}|^2$ with $\alpha\in[0,1/2)$ are \emph{not} monotones under MIO.

Turning now to DIO, we obtain new monotones by replacing the set $\mc{I}$ in the minimizations of Eqns. \eqref{Eq:Renyi-DIO-1} and \eqref{Eq:Renyi-DIO-2} by different sets.  First, by taking the singleton $\{\Delta(\rho)\}$ for each $\rho$, the analog to Eq. \eqref{Eq:Renyi-DIO-1} becomes
\begin{align}
C_{\Delta,\alpha}(\rho)&:=\tfrac{1}{\alpha-1}\log\tr[\rho^\alpha\left(\Delta(\rho)\right)^{1-\alpha}]\;\;\;\;\alpha\in[0,2].\notag
\end{align}
For a pure states $\ket{\psi}=\sum_i\sqrt{p_i}\ket{i}$, this expression yields $C_{\Delta,2-\alpha}(\psi)=S_\alpha(p)$.  Since this is a DIO monotone for $\alpha\in[0,2]$, when combined with Eq. \eqref{Eq:Renyi-Pure} and the fact that DIO $\subset$ MIO, it implies that all R\'{e}nyi entropies $S_\alpha(p_i)$ for pure states are DIO monotones.  This is in sharp contrast to MIO for which monotonicity only holds when $\alpha\geq 1/2$, as shown in Theorem \ref{Thm:MIO-pure}.

In addition, one can define the quantity
\[C_{\Delta,\alpha}^{(q)}(\rho):=\min_{\sigma\in A_\rho}D_\alpha^{(q)}(\rho||\sigma),\;\;\;\; \alpha\in[1/2,\infty],\]
where \\$A_\rho=\left\{\frac{(1+t)\Delta(\rho)-\rho}{t}\;\Big |\;t>0\;;\;\;(1+t)\Delta(\rho)-\rho\geq 0\right\}.$  It turns out that $C_{\Delta,\alpha}^{(q)}(\rho)$ is also a DIO monotone.  Similar to $C_{\alpha}^{(q)}(\rho)$, taking the limit $\alpha\to\infty$ for $C_{\Delta,\alpha}^{(q)}(\rho)$ yields $\log(1+C_{\Delta,R})$, where $C_{\Delta,R}$ is  a new type of robustness measure that we call the $\Delta$-Robustness of Coherence:
 \begin{align}
C&_{\Delta,R}(\rho)=\min_{t\geq 0}\left\{t\;\Big|\;\frac{\rho+t\sigma}{1+t}\in\mI,\;\sigma\geq 0,\;\Delta(\sigma-\rho)=0\right\}.\notag
 \end{align}
For qubits, $C_{\Delta,R}(\rho)$ is also a monotone under MIO, and as shown in the Appendix, the robustness measures $C_R(\rho)$ and $C_{\Delta,R}(\rho)$ completely characterize incoherent transformations of qubit states.
\begin{theorem}
For qubit state $\rho$ and $\sigma$, the transformation $\rho\to\sigma$ is possible by either SIO, DIO, IO, or MIO if and only if both $C_R(\rho)\geq C_R(\sigma)$ and $C_{\Delta,R}(\rho)\geq C_{\Delta, R}(\sigma)$.
\end{theorem}

In conclusion, we have introduced the class of PIO as a physically consistent resource theory of quantum coherence.  Because of PIO's sharply limited abilities, it is desirable to enlarge the free operations to include SIO, IO, DIO, or MIO.  This desire may even be experimentally motivated if one is not be concerned with physical implementations, but instead just wants to know what can be accomplished with a ``black box'' that performs SIO, IO, DIO, or MIO.  We have introduced new monotones for these classes and shown them to all be equivalent for qubit systems.  

A large number of additional results are presented in the Appendix.  In particular, we show that the incoherent Schmidt rank of a pure state is a monotone for SIO/IO/DIO, while it can be increased arbitrarily large by MIO.  The majorization criterion for pure-state transformations is shown to hold true for SIO.  On the other hand, we identify mistakes in the published proof for the claim that majorization likewise characterizes transformations by IO \cite{Du-2015a}.  We also comment on asymmetry-based approaches to quantum coherence and develop the resource theory of $N$-asymmetry.  In the setting of $N$-asymmetry, we find necessary and sufficient conditions for single-copy state transformations. Somewhat surprisingly, these conditions are very similar to the ones obtained in the resource theory of athermality in the limit of zero temperature~\cite{Varun}.  Finally a list of open problems in the resource theory of coherence is given.

\medskip

\emph{Note Added:---}  In the preparation of this letter we became aware of independent work by Marvian and Spekkens \cite{RobIman}, where the physical meaning of incoherent operations is analyzed and the class of dephasing-covariant incoherent operations is presented.

\emph{Acknowledgments:---}
E.C. is supported by the National Science
Foundation (NSF) Early CAREER Award No. 1352326.  G.G. research is supported by NSERC.

\bigskip

\begin{center}
{\Large Appendix}
\end{center}

\tableofcontents

\begin{figure}[b]
\includegraphics[scale=0.45]{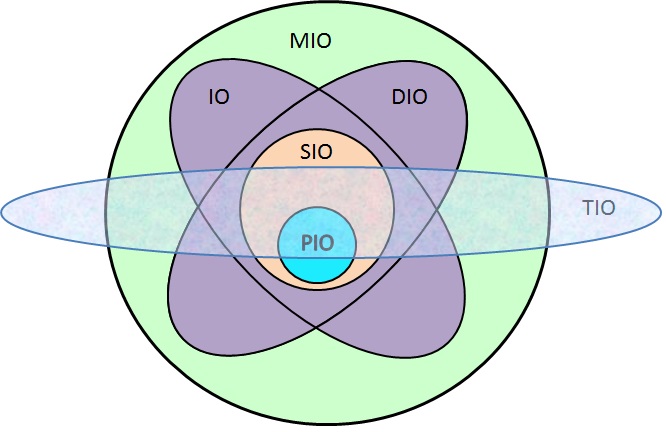}
\caption{\label{MIO}  Heuristic comparison between the 5 incoherence operations MIO/DIO/IO/SIO/PIO and the TIO. Clearly, PIO $\subset$ SIO, SIO $\subset$ IO, and SIO $\subset$ DIO. We also have  IO $\cup$ DIO $\subset$ MIO.
TIO is fundamentally different since the allowed operations in this class depend on the generator $H$ of translations.  In general, PIO will not be subset of TIO since PIO includes permutations. On the other hand, TIO $\not\subset$ MIO because TIO allows for decoherence-free subspaces.}
\end{figure}

\section{Five types of incoherent operations}

We introduce below five types of incoherence operations (IO): (1) Physically IO (PIO is introduced for the first time here)
(2) Strict IO (SIO was very recently introduced in~\cite{Yadin-2015a}, and here, among other things, we fix the error given for the form of the Kraus operators in~\cite{Yadin-2015a})) (3) IO (introduced first in~\cite{Baumgratz-2014a}) (4) Dephasing-covariant IO (DIO is introduced for the first time here) and finally (5) Maximal IO (MIO introduced initially in~\cite{Aberg-2004a} and is developed here). All these 5 types of IO are not completely independent as can be seen in Fig~1. Fig~1 demonstrates heuristically the facts that PIO $\subset$ SIO (i.e. PIO is a subset of SIO) and SIO $\subset$ IO as well as SIO $\subset$ DIO. Furthermore, DIO is not a subset of IO, and IO is not a subset of DIO. Both IO and DIO are subsets of MIO. 

There is another set of incoherent operations, called Translation Invariant Operations (TIO), which emerges in the resource theory of asymmetry.  These operations correspond yet to another type of coherence
that we discuss in Section~\ref{TIO}. A more comprehensive discussion on the distinction between 
TIO and the family of operations PIO/SIO/IO/DIO/MIO can be found in~\cite{RobIman}.

\subsection{Physical Incoherent Operations (PIO)}

We begin by characterizing the general form of a PIO map.  Recall that such a channel can be obtained by performing an incoherent unitary $U_{AB}$ on the input state $\rho_A$ and some fixed incoherent state $\hat{\rho}_B$, and then performing an incoherent projective measurement on system $B$.  Suppose that $\hat{\rho}_B=\sum_{y}p_y\op{y}{y}$ is an arbitrary incoherent state.  A joint incoherent unitary on $AB$ will take the form
\begin{equation}
U_{AB}=\sum_{xy}e^{i\theta_{xy}}\op{\pi_1(xy)\pi_2(xy)}{xy},
\end{equation}
where $(\pi_1(xy),\pi_2(xy))$ is the output of a permutation $\pi$ applied to $(x,y)$.  Let an incoherent projection $\{P_j\}_j$ be applied to system $B$, where $P_j=\sum_{y'\in S_j}\op{y'}{y'}$ for disjoint sets $S_j$.  Upon obtaining outcome $j$, the state of system $A$ is
\begin{align}
&\sum_y p_y \tr_B[\mbb{I}_A\otimes P_jU_{AB}(\rho^A\otimes\op{y}{y})U_{AB}^\dagger]\notag\\
&=\sum_yp_y\sum_{y'\in S_j} \notag\\
&\times\sum_{x,x':\atop \pi_2(xy)=\pi_2(x'y)=y'}e^{i(\theta_{xy}-\theta_{x'y})}\op{\pi_1(xy)}{x}\rho_A\op{x'}{\pi_1(x'y)}\notag\\
&=\sum_y p_y\sum_{y'\in S_j} U^{(y)}_{y'}P^{(y)}_{y'}\rho_A P^{(y)}_{y'}(U^{(y)}_{y'})^\dagger
\end{align}
where
\begin{equation}
U^{(y)}_{y'}=\sum_{x\atop \pi_2(xy)=y'}e^{i\theta_{x y}}\op{\pi_1(xy)}{x}+\sum_{x\atop \pi_2(xy)\not=y'}\op{\hat{\pi}^{(y)}_{y'}(x)}{x}
\end{equation}
for some suitably chosen permutation $\hat{\pi}^{(y)}_{y'}$ such that $U^{(y)}_{y'}$ is unitary, and
\begin{align}
P^{(y)}_{y'}&=\sum_{x:\atop \pi_2(x,y)=y'}\op{x}{x}.\notag
\end{align}
Notice that for each fixed $y$, the projectors $P^{(y)}_{y'}$ are orthogonal and satisfy $\sum_{y'}P_{y'}^{(y)}=\mbb{I}_A$.  This leads to the proposition:
\begin{proposition}
\label{Prop:PIO}
A CPTP map $\hat{\mc{E}}$ is a physically incoherent operation if and only if it can be expressed as a convex combination of maps each having Kraus operators $\{K_j\}_{j=1}^r$ of the form
\begin{equation}
\label{Eq:PIO-Kraus}
K_j=U_jP_j=\sum_{x}e^{i\theta_x}\op{\pi_j(x)}{x}P_{j},
\end{equation}
where the $P_j$ form an orthogonal and complete set of incoherent projectors on system $A$ and $\pi_j$ are permutations.
\end{proposition}

\subsubsection{State Transformations}

Proposition \ref{Prop:PIO} shows that there is very little freedom in the allowable Kraus operators for a PIO map. The following lemma completely characterizes pure state transformations by PIO.
\begin{proposition}
For any two state $\ket{\psi}$ and $\ket{\phi}$, the transformation $\ket{\psi}\to\ket{\phi}$ is possible by PIO if and only if 
\begin{equation}
\label{Eq:PIO-state-transformation-form}
\ket{\psi}=\sum_{i=1}^k \sqrt{p_i} U_i\ket{\phi},
\end{equation}
where the $U_i$ are incoherent isometries such that $P_iU_i\ket{\phi}=U_i\ket{\phi}$ for an orthogonal and complete set of incoherent projectors $\{P_i\}_i$.
\end{proposition} 
\begin{proof}
Necessity of this condition follows from the form of $K_j$ as given in Eq. \eqref{Eq:PIO-Kraus}.  Since $K_j\ket{\psi}\propto\ket{\phi}$ for every $j$, we must have $\frac{1}{\sqrt{p_j}}U_jP_j\ket{\psi}=\ket{\phi}$.  Thus, 
\[\frac{1}{\sqrt{p_j}}P_j\ket{\psi}=U_j^\dagger\ket{\phi}=P_jU_j^\dagger\ket{\phi}.\]  Sufficiency of Eq. \eqref{Eq:PIO-state-transformation-form} can likewise be seen.  Given the form of Eq. \eqref{Eq:PIO-state-transformation-form}, one performs the incoherent projection $\{P_i\}_i$ on $\ket{\psi}$.  Since $P_jU_i\ket{\phi}=0$ for $i\not=j$, outcome $P_j$ renders the post-measurement state $U_j\ket{\phi}$.  The transformation is complete by applying $U_j^\dagger$.
\end{proof}

A generic state $\ket{\psi}$ will not have a decomposition given by Eq. \eqref{Eq:PIO-state-transformation-form} for $k>1$.  Thus, most pure states cannot be transformed into any other outside of their respective incoherent unitary equivalence class.  This situation is highly reminiscent of multipartite entanglement in which most pure states cannot be transformed to any another other outside their respective LU equivalence class.

In the asymptotic setting of many copies, the power of PIO is greatly improved.  The following proposition shows that PIO is just as powerful as Maximally Incoherent Operations (MIO) in terms of distilling maximally coherent bits $\ket{+}=\sqrt{1/2}(\ket{0}+\ket{1})$ from many copies of a pure state.  The optimal distillation rate under MIO is given by $S[\Delta(\psi)]$, where $S[\rho]=-tr[\rho\log\rho]$ is the von Neumann entropy \cite{Winter-2015a}.  
\begin{proposition}[\cite{Winter-2015a}]
\label{Prop:Winter}
For any $\epsilon>0$ and $n$ sufficiently large, the transformation $\ket{\psi}^{\otimes n}\to\overset{\epsilon}{\approx}\ket{+}^{\otimes \lfloor nR\rfloor}$ is possible by PIO whenever $R<S[\Delta(\psi)]$.
\end{proposition}
\begin{proof}
The proof for this is presented in Theorem 3 of Ref. \cite{Winter-2015a} where the authors consider distillation using more general Incoherent Operations (IO).  However, their protocol consists of incoherent unitaries and projections, and therefore it can be accomplished using PIO.
\end{proof}

Rather surprisingly, the reverse transformation $\ket{+}^{\otimes m}\to\overset{\epsilon}{\approx}\ket{\psi}^{\otimes n}$ is not possible for any coherent state $\ket{\psi}$ that is not maximally coherent, i.e. if $\ket{\psi}$ is not of the form $\frac{1}{\sqrt{d}}\sum_{x=1}^de^{i\theta_x}\ket{x}$.  As described in the main text, a proof of this fact follows from communication complexity results in LOCC entanglement transformations.  The key idea is that a PIO transformation $\rho \to\sum_jp_j\rho_j\otimes \op{j}{j}$ can be converted into a bipartite LOCC transformation $\rho^{(mc)}\to\sum_jp_j\rho^{(mc)}_j\otimes \op{jj}{jj}$ with no communication, where 
\begin{equation}
\label{Eq:MC}
\rho=\sum_{xy}c_{xy}\op{x}{y}\quad\Leftrightarrow\quad\rho^{(mc)}=\sum_{xy}c_{cy}\op{xx}{yy},
\end{equation}
 and likewise for the $\rho_j\Leftrightarrow\rho_j^{(mc)}$.  Specifically, if $\{U_jP_j\}$ is the PIO measurement, then the corresponding LOCC protocol consists of Alice locally measuring $\{U_jP_j\}$, Bob learning the outcome of this measurement through the projective measurement $\{P_j\}$, and then him applying the corresponding $U_j$.  Therefore, if $\op{+}{+}^{\otimes m}\to\sum_jp_j\op{\psi_j}{\psi_j}\otimes \op{j}{j}$ by PIO with $\sum_j p_j\op{\psi_j}{\psi_j}\overset{\epsilon}{\approx}\op{\psi}{\psi}^{\otimes n}$ for arbitrarily small $\epsilon$ and $m$ is sufficiently large, then it is possible to transform sufficiently large copies of an EPR state arbitrarily close to $\ket{\psi^{(mc)}}^{\otimes n}$ by local operations and no communication.  However, as proven in Refs. \cite{Hayden-2003a, Harrow-2003a}, for any fixed $n$, there exists an $\epsilon$-dependent lower bound on the communication needed to perform such an entanglement dilution, provided $\ket{\psi^{(mc)}}$ is not maximally entangled or a product state.  

From this result we see that maximally coherent states are the weakest among all pure states, in terms of their ability to transform into other states.  Under asymptotic PIO, the entire hierarchy of coherent states gets turned upside down

\subsection{Strictly Incoherent Operations (SIO)}
\begin{definition}
Let $\mE^{A\to B}:\mL(\mH^A)\to\mL(\mH^B)$ be a CPTP map. Then, $\mE^{A\to B}$ is said to be a Strictly Incoherent Operation (SIO) if it can be represented by Kraus operators $\{M_j\}$ such that
\be
\Delta\left(M_j\rho M_j^\dagger\right)=
M_j\Delta(\rho)M_j^\dagger\;\;\;\forall j,\;\forall\;\rho\;.
\ee
\end{definition}
\begin{lemma}
\label{Lem:SIO-Characterization}
Let $\mE^{A\to B}:\mL(\mH^A)\to\mL(\mH^B)$ be a CPTP map. Then, $\mE^{A\to B}$ is SIO if and only if it can be represented by Kraus operators $\{M_j\}$ of the form
\begin{equation}
\label{Eq:SIO-Kraus-Form}
M_j=\sum_{x=1}^{d_A}c_{jx}\op{\pi_j(x)}{x}.
\end{equation}
\end{lemma}
\begin{proof}
Sufficiency is obvious to check.  Suppose now that $\mE^{A\to B}$ is SIO.  Following same arguments of  Lemma \ref{Lem:DIO-Characterization}, there must exist Kraus operators $\{M_j\}$ with the properties that
\begin{align}
\Delta\left(M_j\op{x}{x} M_j^\dagger\right)&=M_j\op{x}{x} M_j^\dagger\;\;\text{ and}\label{cond1SIO}\\
\Delta \left(M_j |x\lr x'| M_j^\dagger\right)&=0\;\label{cond2SIO}
\end{align}
for all $x',x\in\{1,...,d_A\}$ with $x'\neq x$.  Eq. \eqref{cond1SIO} implies that
\begin{equation}
M_j=\sum_{x=1}^{d_A}c_{j,x}\op{f_j(x)}{x},
\end{equation}
where $f_j:\{1,\cdots,d_A\}\to\{1,\cdots,d_A\}$.  Eq. \eqref{cond2SIO} implies that 
\begin{equation}
\bra{y}M_j\op{x}{x'}M_j^\dagger\ket{y}=0\;\;\;\forall x,x',y,
\end{equation}
which is equivalent to the condition that $f_j$ is one-to-one.  Thus, $f_j$ is a permutation $\pi_j$ and $M_j$ takes the form of Eq. \eqref{Eq:SIO-Kraus-Form}.
\end{proof}

In Ref. \cite{Yadin-2015a} it is claimed that if a set of Kraus operators $\{M_j\}$ represents a strictly incoherent operation, then it has a physical implementation described by
\begin{align}
\rho_A\to &\sum_{j}M_j\rho_A M_j^\dagger\otimes\op{j}{j}\notag\\
=& \sum_j {}_B\bra{j}U_{AB}(\rho_A\otimes\op{0}{0})U_{AB}^\dagger\ket{j}_B \otimes\op{j}{j},
\end{align}
where $U_{AB}$ is a Stinespring dilation of the form
\begin{equation}
U_{AB}=\sum_{j} \op{\pi(j)}{j}_A\otimes \op{\psi_j}{0}_B
\end{equation}
and $\pi$ is a permutation and $\ket{\psi_i}$ is an arbitrary state.  However, this form is not general enough to fully characterize SIO Stinespring dilations.
\begin{proposition}
Let $\mE^{A\to B}:\mL(\mH^A)\to\mL(\mH^B)$ be a CPTP map. Then, $\mE^{A\to B}$ is SIO if and only if it has a Stinespring dilation
\begin{equation}
U_{AB}=\sum_{j,x} c_{jx}\op{\pi_j(x)}{x}\otimes \op{j}{0}.
\end{equation}
\end{proposition}
\begin{proof}
This follows directly from Lemma \ref{Lem:SIO-Characterization} and the fact that a general Stinespring dilation of $\mE^{A\to B}$ for SIO Kraus operators $\{M_j\}$ can be written as
\begin{align}
U_{AB}&=\sum_{j} M_j\otimes \op{j}{0}\notag\\
&=\sum_{j}\sum_{x=1}^{d_A}c_{jx}\op{\pi_j(x)}{x}\otimes\op{j}{0}.
\end{align}
\end{proof}

\subsubsection{Relating SIO to Maximally Correlated LOCC}
The discussion after Proposition \ref{Prop:Winter} describes how every PIO operation can be translated into a zero communication LOCC protocol.  A similar relationship holds for SIO and one-way LOCC. 
\begin{proposition}
\label{Prop:SIO-LOCC}
Using the notation of Eq. \eqref{Eq:MC}, if $\rho\to\sigma$ by SIO, then there exists a bipartite LOCC transformation $\rho^{(mc)}\to\sigma^{(mc)}$.
\end{proposition}
\begin{proof}
Let $\{M_j\}$ be a set of SIO Kraus operators so that for state $\rho=\sum_{xy}d_{xy}\op{x}{y}$ the QC post-measurement state is
\begin{align}
\sigma&=\sum_j M_j\rho M_j^\dagger\otimes\op{j}{j}\notag\\
&=\sum_{x,y}c_{jx}c_{jy}^*d_{xy}\op{\pi_j(x)}{\pi_j(y)}\otimes\op{j}{j},
\end{align}
where we have used Eq. \eqref{Eq:SIO-Kraus-Form}.  Then the transformation $\rho^{(mc)}\to\sigma^{(mc)}$ can be accomplished by Alice performing the measurement $\{M_j\}$, announcing her result ``$j$'' to Bob, and then Bob performing the local permutation $\Pi_j:\ket{x}\to\ket{\pi_j(x)}$.
\end{proof}

\subsubsection{State Transformations}

Using Proposition \ref{Prop:SIO-LOCC}, we can completely classify pure state transformations under SIO.  The following is an analog to Nielsen's theorem for entanglement transformations of bipartite pure states \cite{Nielsen-1999a}. Consider two states 
\begin{align*}
\ket{\psi}&=\sum_{i=1}^m\sqrt{\psi_i^\downarrow}\ket{i},&\ket{\phi}&=\sum_{i=1}^n\sqrt{\phi_i^\downarrow}\ket{i}
\end{align*}
where we have assumed without loss of generality that the $\psi_i^\downarrow$ are non-negative and ordered such that $\psi_i^\downarrow\geq \psi_{i+1}^\downarrow$, and likewise for the $\phi_i^\downarrow$.  We say that $\ket{\phi}$ majorizes $\ket{\psi}$ (denoted by $\vec{\tau}(\psi)\prec\vec{\tau}(\phi)$) if $\sum_{i=1}^{k}\psi_i^\downarrow\leq\sum_{i=1}^k\phi_i^\downarrow$
for all $k=1,\cdots,\max\{m,n\}$, where a sufficient number of zeros are padded to the vector of shorter length so that both summations can be taken over $\max\{m,n\}$ elements.
\begin{lemma}\label{Lem:SIO-Majorization}
The state transformation $\ket{\psi}\to\ket{\phi}$ is possible by SIO iff $\vec{\tau}(\psi)\prec\vec{\tau}(\phi)$.
\end{lemma}
\begin{proof}
Sufficency:  Suppose that $\vec{\tau}(\psi)\prec\vec{\tau}(\phi)$.  Then there exists a doubly stochastic matrix $D$ such that $\vec{\tau}(\psi)=D\vec{\tau}(\phi)$ \cite{Bhatia-2000a}.  Birkhoff's Theorem assures that $D=\sum_\alpha p_\alpha \Pi_\alpha$, where the $p_\alpha$ form a probability distribution and the $\Pi_\alpha$ are permutation matrices.  Then define the operators $M_\alpha:=\sqrt{p_\alpha}\Pi_\alpha\bullet S$, where the elements of $S$ are given by $[[S]]_{ij}=\sqrt{\phi_i}/\sqrt{\psi_j}$ and ``$\bullet$'' denotes the Hadamard product.  Recall that the Hadamard product of two matrices $A$ and $B$ is the matrix $A\bullet B$ with elements $[[A\bullet B]]_{ij}=[[A]]_{ij}[[B]]_{ij}$.  Note that each $M_\alpha$ has the form of Eq. \eqref{Eq:SIO-Kraus-Form}.  By construction $M_\alpha\otimes\Pi_\alpha \ket{\psi}\propto\ket{\phi}$ for every $\alpha$, and the relation $\vec{\tau}(\psi)=\sum_\alpha p_\alpha \Pi_\alpha\vec{\tau}(\phi)$ readily implies that $\sum_\alpha M_\alpha^\dagger M_\alpha=\mbb{I}$. \\ Necessity:  Now suppose that $\ket{\psi}\to\ket{\phi}$ by SIO.  By Prop. \ref{Prop:SIO-LOCC}, this means that $\ket{\psi^{(mc)}}\to\ket{\phi^{(mc)}}$ by bipartite LOCC.  However, a necessary condition for this is that $\vec{\tau}(\psi)\prec\vec{\tau}(\phi)$ \cite{Nielsen-1999a}.
\end{proof}
By the same arguments, additional statements about SIO pure-state transformations can be made that are analogous to statements in bipartite LOCC.  The following are the coherence versions of the results presented in \cite{Jonathan-1999a} and \cite{Vidal-1999a} respectively.
\begin{proposition}
The multi-outcome transformation $\ket{\psi}\to\{\ket{\phi_i},p_i\}$ is possible by SIO iff $\vec{\tau}(\psi)\prec\sum_ip_i\vec{\tau}(\phi_i)$.
\end{proposition}
\begin{proposition}
The maximum probability of converting $\ket{\psi}\to\ket{\phi}$ is given by
\begin{equation}
\min_{k\in\{1,\cdots,\max\{m,n\}\}}\frac{\sum_{i=k}^n\psi_i^\downarrow}{\sum_{i=k}^n\phi_I^\downarrow}.
\end{equation}
\end{proposition}

With Lemma \ref{Lem:SIO-Majorization}, the asymptotic transformation of pure states becomes reversible under SIO.  Indeed, the dilution protocol described in Ref. \cite{Winter-2015a} relies on being able to perform any pure state transformation provided the majorization condition is satisfied.  We thus have
\begin{corollary}[\cite{Winter-2015a}]
For any $\epsilon>0$ and $n$ sufficiently large, the transformation $\ket{\psi}^{\otimes \lfloor nR\rfloor}\to\overset{\epsilon}{\approx}\ket{\varphi}^{\otimes n}$ is possible whenever $R<S[\Delta(\psi)]/S[\Delta(\varphi)]$.
\end{corollary}

\subsection{Incoherent Operations}

\begin{definition}
Let $\mE^{A\to B}:\mL(\mH^A)\to\mL(\mH^B)$ be a CPTP map. Then, $\mE^{A\to B}$ is said to be an Incoherent Operation (IO) if it can be represented by Kraus operators $\{M_\alpha\}$ such that
\be
\Delta\left(M_\alpha\op{x}{x} M_\alpha^\dagger\right)=M_\alpha\op{x}{x} M_\alpha^\dagger\;\;\;\;\forall x.
\ee
\end{definition}
From this definition, it is easy to see that an arbitrary incoherent measurement has Kraus operators $\{M_\alpha\}_{\alpha}$ of the form
\begin{equation}
M_\alpha=\sum_{i=1}^dc_{\alpha,i}\op{f_\alpha(i)}{i}
\end{equation}
where $f_{\alpha}:\{1,\cdots,d\}\to\{1,\cdots,d\}$ and the completion identity demands
\begin{align}
\sum_{\alpha\text{ such that} \atop f_\alpha(i)=f_\alpha(j)}c_{\alpha,i}^*c_{\alpha,j}=\delta_{ij}.
\end{align}
Note that we could further decompose the sum as
\begin{equation}
\label{Eq:Completition}
\sum_{\alpha\text{ such that} \atop f_\alpha(i)=f_\alpha(j)}c_{\alpha,i}^*c_{\alpha,j}=\sum_{k=1}^d\sum_{\alpha\text{ such that} \atop i,j\in f^{-1}_\alpha(k)}c_{\alpha,i}^*c_{\alpha,j}=\delta_{ij}.
\end{equation}

Incoherent operations have been studied extensively in the literature, and here we only comment on pure state transformations.  It has been reported that the majorization condition characterizes pure state transformations under IO; i.e. that Lemma \ref{Lem:SIO-Majorization} can be extended to IO \cite{Du-2015a, Renes-2015a}.  However, as we now discuss, the proofs given in these references are not correct.  It is still an open question whether $\vec{\tau}(\psi)\prec\vec{\tau}(\phi)$ is necessary for an IO transformation $\ket{\psi}\to\ket{\phi}$.

\subsubsection{Mistakes in the Majorization Proofs}
Now for a state $\ket{\psi}=\sum_i\psi_i\ket{i}$, let us consider the action
\begin{align}
M_\alpha\ket{\psi}=\sum_{k=1}^d\left(\sum_{i\in f^{-1}(k)}c_{\alpha,i}\psi_i\right)\ket{k}.
\end{align}
What interests us are the diagonal elements of $\Delta\left(\sum_\alpha M_\alpha\op{\psi}{\psi}M_\alpha^\dagger\right)$.  They have undergone the transformation
\begin{align}
\left(|\psi_k|^2\right)_k &\to \left(\sum_\alpha\left(\sum_{i\in f_\alpha^{-1}(k)}c_{\alpha,i}\psi_i\right)\left(\sum_{j\in f_\alpha^{-1}(k)}c_{\alpha,j}^*\psi_j^*\right)\right)_k\notag\\
&=\left(\sum_\alpha\left(\sum_{i,j\in f_\alpha^{-1}(k)}\psi_i\psi_j^*c_{\alpha,i}c_{\alpha,j}^*\right)\right)_k\notag\\
&=\left(\sum_{i,j}\psi_i\psi_j^*\sum_{\alpha\text{ such that} \atop i,j\in f^{-1}_\alpha(k)} c_{\alpha,i}c_{\alpha,j}^*\right)_k.
\end{align}
In Ref. \cite{Du-2015a}, the authors assume that for each value of $k$, the cross terms vanish.  In other words, the assumption is that
\[\sum_{\alpha\text{ such that} \atop i,j\in f^{-1}_\alpha(k)} c_{\alpha,i}c_{\alpha,j}^*=\delta_{ij}\] when, in fact, the full condition is given by Eq. \eqref{Eq:Completition}. 

To bring this out more explicitly, we adopt the notation used in \cite{Du-2015a}.  From the completion identity, Eq. 18 of \cite{Du-2015a} gives
\begin{equation}
\label{Eq:Fix1}
\sum_{n}(\delta_{1,i(2)}\delta_{1,i(3)}+\delta_{2,i(2)}\delta_{2,i(3)})\overline{k_2^{(n)}}k_3^{(n)}=0.
\end{equation}
Note here the authors are assuming that the $\delta_{j,i(l)}$ do depend on $n$, which is not true in general.  Nevertheless, let us momentarily continue with the argument with $\delta_{j,i(l)}$ being independent of $n$.  Because the measurement is incoherent, we have that 
\begin{align}
\delta_{1,i(2)}\delta_{1,i(3)}&\not=0\quad\Rightarrow\quad\delta_{2,i(2)}\delta_{2,i(3)}=0\notag\\
\delta_{2,i(2)}\delta_{2,i(3)}&\not=0\quad\Rightarrow\quad \delta_{1,i(2)}\delta_{1,i(3)}=0.
\end{align}
This means that Eq. \eqref{Eq:Fix1} implies
\begin{equation}
\label{Eq:Fix2}
\sum_{n}\delta_{1,i(2)}\delta_{1,i(3)}\overline{k_2^{(n)}}k_3^{(n)}=\sum_n\delta_{2,i(2)}\delta_{2,i(3)}\overline{k_2^{(n)}}k_3^{(n)}=0.
\end{equation}
Therefore, when computing $\sum_n|\cdot|^2$ in their Eq. 21, the LHS of the second equation becomes
\begin{align}
&\sum_n|\delta_{2,i(2)}k_2^{(n)}\psi_2+\delta_{2,i(3)}k_3^{(n)}\psi_3|^2\notag\\
&=\delta_{2,i(2)}\psi_2^2+\delta_{2,i(3)}\psi_3^2\notag\\
&\qquad+\psi_2\psi_3\sum_n\delta_{2,i(2)}\delta_{2,i(3)}(\overline{k_2^{(n)}}k_3^{(n)}+\overline{k_3^{(n)}}k_2^{(n)})\notag\\
&=\delta_{2,i(2)}\psi_2^2+\delta_{2,i(3)}\psi_3^2,
\end{align}
where we use Eq. \eqref{Eq:Fix2}.  But now let us consider the most general IO measurement by allowing $\delta_{j,i(l)}$ to depend on $n$.  That is, we make the replacement $\delta_{j,i(j)}\to\delta_{j,i(j)}^{(n)}$.  Then Eq. \eqref{Eq:Fix1} becomes
\begin{equation}
\sum_{n}(\delta^{(n)}_{1,i(2)}\delta_{1,i(3)}^{(n)}+\delta_{2,i(2)}^{(n)}\delta_{2,i(3)}^{(n)})\overline{k_2^{(n)}}k_3^{(n)}=0.
\end{equation}
However, we no longer have Eq. \eqref{Eq:Fix2} because of the dependence on $n$.  In other words, in general $\sum_n\delta_{2,i(2)}\delta_{2,i(3)}\overline{k_2^{(n)}}k_3^{(n)}\not=0$.  Therefore,
\begin{align}
&\sum_n|\delta_{2,i(2)}^{(n)}k_2^{(n)}\psi_2+\delta_{2,i(3)}^{(n)}k_3^{(n)}\psi_3|^2\notag\\
&=\sum_n\delta_{2,i(2)}^{(n)}|k_2^{(n)}|^2\psi_2^2+\sum_n\delta_{2,i(3)}^{(n)}|k_3^{(n)}|^2\psi_3^2\notag\\
&\qquad+\psi_2\psi_3\sum_n\delta_{2,i(2)}^{(n)}\delta_{2,i(3)}^{(n)}(\overline{k_2^{(n)}}k_3^{(n)}+\overline{k_3^{(n)}}k_2^{(n)}).
\end{align}
The cross-term no longer vanishes.

An alternative proof for the majorization condition was presented in Ref. \cite{Renes-2015a}.  The proof technique used  is similar to the proof of Lemma \ref{Lem:SIO-Majorization} in which the incoherent transformation is mapped to a bipartite LOCC pure state transformation.  However, the LOCC measurement described in that paper is not trace-preserving, and it is not clear how this can be remedied \cite{Yang-2015a}.

\subsubsection{Majorization for a Special Subclass of IO}

We next introduce yet another class of incoherent operations for which majorization precisely captures pure-state convertibility.
\begin{definition}
Let $\mE^{A\to B}:\mL(\mH^A)\to\mL(\mH^B)$ be a CPTP map. Then, $\mE^{A\to B}$ is said to be a special Incoherent Operation (sIO) if it can be represented by Kraus operators $\{M_\alpha\}$ each having the form
\be
\label{Eq:Kraus-sIO}
M_\alpha=\sum_{x}c_{\alpha x}\Pi_\alpha\op{f(x)}{x},
\ee
where $f:\{1,\cdots,d\}\to\{1,\cdots,d\}$ and $\Pi_\alpha$ is a permutation.  Note that SIO $\subset$ sIO $\subset$ IO.
\end{definition}
We first show that the statement of Proposition \ref{Prop:SIO-LOCC} can be extended to sIO operations.  However, the corresponding LOCC transformation now uses two-way classical communication.
\begin{proposition}
\label{Prop:sIO-LOCC}
If $\rho\to\sigma$ by sIO, then there exists a bipartite LOCC transformation $\rho^{(mc)}\to\sigma^{(mc)}$.
\end{proposition}
\begin{proof}
Suppose that $\rho\to\sigma=\sum_\alpha M_\alpha\rho M_\alpha^\dagger\otimes\op{\alpha}{\alpha}$ for sIO Kraus operators $\{M_\alpha\}$ given by Eq. \eqref{Eq:Kraus-sIO}.  Let $S\subset\{1,\cdots,d\}$ denote the range of $f$, $\kappa_s=|f^{-1}(s)|$ for $s\in S$, and $\kappa=\prod_{s\in S}\kappa_s$.  For each $s\in S$, let $\{\ket{s,j_s}:j_s=0,\cdots,|f^{-1}(s)|-1\}$ be a relabeling of the kets $\ket{x}$ with $x\in f^{-1}(s)$.  Next we want to define a generalized Hadamard basis with respect to the $\ket{s,j_s}$:
\[\bigg\{\ket{\widetilde{s,k_s}}:=\sum_{j_s=0}^{\kappa_s-1}e^{i2\pi j_sk_s/\kappa_s}\ket{s,j_s}\bigg\}_{k_s=0,\cdots,\kappa_s-1}.\]
Finally, for every sequence $\vec{k}=(k_1,k_2,\cdots, k_{|S|})$ with $k_s\in\{0,\cdots,\kappa_s-1\}$, define the operator
\begin{equation}
N_{\vec{k}}=\frac{1}{\sqrt{\kappa}}\sum_{s=1}^{|S|} \op{s}{\widetilde{s,k_s}}.
\end{equation} 
It can be seen that $\sum_{\vec{k}}N_{\vec{k}}^\dagger N_{\vec{k}}=\mbb{I}$.  The LOCC protocol then consists of Bob first performing the measurement $\{N_{\vec{k}}\}_{\vec{k}}$.  The state transformation corresponding to outcome $\vec{k}=(k_s)_{s=1}^{|S|}$ is 
\begin{align}
\rho^{(mc)}&=\sum_{xy}d_{xy}\op{xx}{yy}\notag\\
&=\sum_{ss'}\sum_{j_s,j_{s'}}d_{sj_s,s'j_{s'}}\op{s,j_s}{s',j_{s'}}_A\otimes\op{s,j_s}{s',j_{s'}}_B\notag\\
&\to\propto\sum_{ss'}\sum_{j_s,j_{s'}}d_{sj_s,s'j_{s'}}e^{i 2\pi (j_s-j_{s'})k_s/\kappa_s}\notag\\
&\qquad\qquad\qquad\times \op{s,j_s}{s',j_{s'}}_A\otimes\op{s}{s'}_B.
\end{align}
Bob then announces his outcome $\vec{k}=(k_s)_{s=1}^{|S|}$ to Alice who subsequently performs the unitary
\begin{equation}
U_{\vec{k}}=\sum_{s}\sum_{j_s}e^{-i2\pi j_s k_s/\kappa_s}\op{s,j_s}{s,j_s}.
\end{equation}
At this stage, Alice and Bob share the state
\begin{equation}
\hat{\rho}^{(mc)}=\sum_{xy}d_{x,y}\op{x}{y}_A\otimes\op{f(x)}{f(y)}_B,
\end{equation}
regardless of Bob's outcome $\vec{k}$.  Alice now locally performs the sIO measurement $\{M_\alpha\}$.  She announces her result to Bob who then performs the conditional permutation $\Pi_\alpha$ on his system.  Thus, the resulting QC state is
\begin{align}
&\sigma^{(mc)}=\notag\\
&\sum_{xy}d_{x,y}c_{\alpha,x}c_{\alpha,y}^*\notag\\
&\qquad\times(\Pi_\alpha\otimes\Pi_\alpha)\op{f(x)f(x)}{f(y)f(y)}_{A_1B_1}(\Pi_\alpha\otimes\Pi_\alpha)\notag\\
&\qquad\qquad\otimes\op{\alpha\alpha}{\alpha\alpha}_{A_2B_2}.
\end{align}
\end{proof}

\begin{corollary}\label{Cor:sIO-Majorization}
The state transformation $\ket{\psi}\to\ket{\phi}$ is possible by sIO iff $\vec{\tau}(\psi)\prec\vec{\tau}(\phi)$.
\end{corollary}

\subsection{Dephasing-covariant Incoherent Operations (DIO)}

\begin{definition}
Let $\mE^{A\to B}:\mL(\mH^A)\to\mL(\mH^B)$ be a CPTP map. Then, $\mE^{A\to B}$ is said to be a Dephasing-Covariant Incoherent operation (DIO) if
\be
[\Delta,\mE^{A\to B}]=0
\ee
which is equivalent to
\be
\Delta\left(\mE^{A\to B}(\rho)\right)=
\mE^{A\to B}\left(\Delta(\rho)\right)\;\;\forall\;\rho\;.
\ee
\end{definition}

\begin{lemma}
\label{Lem:DIO-Characterization}
Let $\mE^{A\to B}:\mL(\mH^A)\to\mL(\mH^B)$ be a CPTP map. Then, $\mE^{A\to B}$ is DIO if and only if
for all $x',x\in\{1,...,d_A\}$ with $x'\neq x$:
\begin{align}
&\mE^{A\to B}(|x\lr x|)\in \mI\;\;\text{ and}\label{cond1}\\
&\Delta \left(\mE^{A\to B}(|x\lr x'|)\right)=0\;.\label{cond2}
\end{align}
\end{lemma}
\begin{proof}
The first condition in the equation above ensures that $\mE^{A\to B}$ is a MIO.
Therefore this is a necessary condition. The second condition is also necessary since
\begin{align*}
&\Delta \left(\mE^{A\to B}(|x\lr x'|)\right)=\\
&\mE^{A\to B}\left(\Delta(|x\lr x'|)\right)=\mE^{A\to B}\left(0\right)=0
\end{align*}

Now, to see that these two conditions are sufficient, note that 
any density matrix $\rho$ acting on $\mH^A$
can be decomposed as 
\be
\rho=\Delta(\rho)+Z
\ee
where $Z$ is an Hermitian matrix with zeros on the diagonal.
We therefore have
\begin{align}
\Delta\left(\mE^{A\to B}(\rho)\right) & =\Delta\left(\mE^{A\to B}(\Delta(\rho))\right)+\Delta\left(\mE^{A\to B}(Z)\right)\nonumber\\
& =\mE^{A\to B}(\Delta(\rho))+\Delta\left(\mE^{A\to B}(Z)\right)\nonumber\\
&=\mE^{A\to B}(\Delta(\rho))
\label{DIOMIO}
\end{align}
where the second equality follows from~\eqref{cond1}, and the third equality follows from~\eqref{cond2}. Hence, $\mE^{A\to B}$ is DIO iff~\eqref{cond1} and~\eqref{cond2} holds. 
\end{proof}

\begin{proposition}
Let $\mE$ be a CPTP map that is DIO. Then, its dual $\mE^*$ is a unital CP map satisfying
$[\mE^*,\Delta]=0$ (i.e. $\mE^*$ is also DIO but not necessarily trace preserving).
\end{proposition}
\begin{proof}
For $x\neq x'$ we get
\begin{align*}
\la y\left(\mE^{*}(|x\lr x'|)\right)|y\ra &=\tr\left[|y\lr y|\mE^{*}(|x\lr x'|)\right]\\
&=\tr\left[\mE(|y\lr y|)|x\lr x'|\right]\\
&=\la x\left(\mE(|y\lr y|)\right)|x'\ra=0
\end{align*}
since $\mE$ is DIO. We therefore showed that $\Delta(\mE^{*}(|x\lr x'|))=0$ for $x\neq x'$.
Similarly, for $y\neq y'$ we get
\begin{align*}
\la y\left(\mE^{*}(|x\lr x|)\right)|y'\ra &=\tr\left[|y'\lr y|\mE^{*}(|x\lr x|)\right]\\
&=\tr\left[\mE(|y'\lr y|)|x\lr x|\right]\\
&=\la x\left(\mE(|y'\lr y|)\right)|x\ra=0
\end{align*}
since $\mE$ is DIO (we used $\Delta\mE(|y'\lr y|)=0$ from Lemma~\ref{Lem:DIO-Characterization}).
Therefore, from Lemma~\ref{Lem:DIO-Characterization} above it follows that $[\mE^*,\Delta]=0$.
\end{proof}

Note that if we denote by
\be\label{vnotation}
\mbf{v}_{y|x}\equiv\begin{pmatrix}
\la y|M_1|x\ra \\
\la y|M_2|x\ra \\
\vdots\\
\la y|M_m|x\ra
\end{pmatrix}\in\mathbb{C}^m\;,
\ee
we get the following corollary:

\begin{corollary}
Using the notation of~\eqref{vnotation}, a CPTP map $\mE^{A\to B}:\mL(\mH^A)\to\mL(\mH^B)$
is a DIO if and only if there exists conditional probabilities $r_{y|x}$ such that
\begin{align}
& \mbf{v}_{y'|x}^{\dag}\mbf{v}_{y|x}=r_{y|x}\delta_{yy'}\label{first1}\\
& \mbf{v}_{y|x}^{\dag}\mbf{v}_{y|x'}=r_{y|x}\delta_{xx'}\label{second1}\;.
\end{align}
\end{corollary}

Consider now the equation $\sigma=\mE(\rho)$ where $\mE$ is DIO. We therefore have
\be
\sigma_{yy'}=\sum_{x,x'}\rho_{xx'}\la y|\mE(|x\lr x'|)|y'\ra
\ee
In the notations above, this is equivalent to
\be
\sigma_{yy'}=\sum_{x,x'}\rho_{xx'}\mbf{v}_{y'|x'}^{\dag}\mbf{v}_{y|x}
\ee
The diagonal terms have the form
\be
\sigma_{yy}=\sum_{x}r_{y|x}\rho_{xx}\;.
\ee


\subsection{Maximal Incoherent Operations (MIO)}

\begin{definition}
Let $\mE^{A\to B}:\mL(\mH^A)\to\mL(\mH^B)$ be a CPTP map. Then $\mE^{A\to B}$ is a Maximal Incoherent Operation (MIO) if
\be
\Delta\circ \mE^{A\to B}\circ \Delta=\mE^{A\to B}\circ \Delta\;.
\ee
\end{definition}

Let $\mE^{A\to B}:\mL(\mH^A)\to\mL(\mH^B)$ be a CPTP map with an operator sum representation $\{M_{j}\}_{j=1}^{m}$, and 
let $\mM$ denotes the set of MIOs. Then from the definition above, $\mE^{A\to B}\in\mM$ if and only if 
\be\label{maxi}
\sum_{j=1}^{m}\la y|M_j|x\ra\la x|M_{j}^{\dag}|y'\ra=0
\ee
for all $x\in\{1,...,d_A\}$ and $y\neq y'$ with $y,y'\in\{1,...,d_B\}$.  Using the notation of~\eqref{vnotation}
we get that then $\mE^{A\to B}\in\mM$ if and only if there exists $d_Ad_B$ vectors $\mbf{v}_{y|x}\in\mathbb{C}^m$,
and conditional probability distribution $r_{y|x}$ (i.e. $r_{y|x}\geq 0$ and $\sum_yr_{y|x}=1$) such that
\begin{align}
& \mbf{v}_{y'|x}^{\dag}\mbf{v}_{y|x}=r_{y|x}\delta_{yy'}\label{first}\\
& \sum_{y=1}^{d_B}\mbf{v}_{y|x}^{\dag}\mbf{v}_{y|x'}=\delta_{xx'}\label{second}\;,
\end{align}
where the first equation follows from~\eqref{maxi} and the second from $\sum_jM_{j}^{\dag}M_j=I$.

\subsubsection{Pure state transformations}

Consider a MIO that convert $|\psi\ra=\sum_{x}\sqrt{p_x}|x\ra$ to $|\phi\ra=\sum_{y}\sqrt{q_y}|y\ra$.
In this case, we have $|\phi\lr\phi|=\mE(|\psi\lr\psi|)$, where $\mE$ is MIO. 
 Then, there must exists coefficients $c_j$ such that $\sum_{j=1}^{m}|c_j|^2=1$ and  $M_j|\psi\ra=c_j|\phi\ra$. Denoting $\mbf{c}\equiv(c_j)_j\in\mathbb{C}^m$ gives
\be\label{pure}
\sqrt{q_y}\mbf{c}=\sum_x\sqrt{p_x}\mbf{v}_{y|x}\;\;\;\;\forall\;y\;.
\ee
Consider now the simpler case of $d_A=2$. We will also assume that $q_y>0$ and $d_B\geq 3$. The case $d_B=2$ is a specially case of the qubit mixed state transformation to be discussed later.
Denote by $r_y\equiv r_{y|0}$ and $t_y\equiv r_{y|1}$ the two probability distributions, and denote also $\mbf{v}_{y|0}\equiv\mbf{v}_y$ and $\mbf{v}_{y|1}\equiv\mbf{u}_y$. With these notations,
conditions~\eqref{first},~\eqref{second},~\eqref{pure} take the form:
\begin{align}
& \mbf{v}_{y}^{\dag}\mbf{v}_{y'}=r_y\delta_{yy'} \;\;,\;\;\mbf{u}_{y}^{\dag}\mbf{u}_{y'}=t_y\delta_{yy'}\nonumber\\
&\sum_{y=1}^{d_B}\mbf{v}_{y}^{\dag}\mbf{u}_y=0\;\;,\;\;\sqrt{q_y}\mbf{c}=\sqrt{p_0}\mbf{v}_y+\sqrt{p_1}\mbf{u}_y\label{simp}
\end{align}
The last equation can be written as:
\be
\sqrt{p_1}\mbf{u}_y=\sqrt{q_y}\mbf{c}-\sqrt{p_0}\mbf{v}_y
\ee
Hence, we must have
\begin{align}
& p_1t_y\delta_{yy'}=p_1\mbf{u}_{y}^{\dag}\mbf{u}_{y'}=\nonumber\\
&\sqrt{q_yq_{y'}}+p_0r_y\delta_{yy'}-\sqrt{p_0}\left(\sqrt{q_y}\mbf{c}^{\dag}\mbf{v}_{y'}+\sqrt{q_{y'}}\mbf{v}_{y}^{\dag}\mbf{c}\right)
\end{align}
where we have used the normalization of $\mbf{c}$ and the orthogonality of $\{\mbf{v}_y\}$ and of $\{\mbf{u}_y\}$. Therefore, after dividing both sides of the equation by $\sqrt{q_yq_{y'}}$ (which is non-zero) we get
\be\label{ynoty}
1=\sqrt{p_0}\left(\frac{\mbf{c}^{\dag}\mbf{v}_{y'}}{\sqrt{q_{y'}}}+\frac{\mbf{v}_{y}^{\dag}\mbf{c}}{\sqrt{q_y}}\right)\;\forall\;y\neq y'
\ee
and for $y=y'$
\be\label{yisy}
p_1t_y=q_y+p_0r_y-\sqrt{p_0q_y}\left(\mbf{c}^{\dag}\mbf{v}_{y}+\mbf{v}_{y}^{\dag}\mbf{c}\right)
\ee
From~\eqref{ynoty} we get that
\be
\sqrt{p_0}\frac{\mbf{v}_{y}^{\dag}\mbf{c}}{\sqrt{q_y}}\equiv a
\ee
where $a$ is some complex number independent of $y$ satisfying $a+\bar{a}=1$.
Substituting this into~\eqref{yisy} we get
\be
p_1t_y=q_y+p_0r_y-q_y
\ee
This equation holds iff
\be
p_0=p_1=\frac{1}{2}\;\;,\;\;\text{ and }\;\;t_y=r_y\;.
\ee
With these choices, the first equation of~\eqref{simp} gives
\be
0=\sum_{y=1}^{d_B}\mbf{v}_{y}^{\dag}\mbf{u}_y=\sum_{y=1}^{d_B}\mbf{v}_{y}^{\dag}\left(\sqrt{2q_y}\mbf{c}-\mbf{v}_y\right)
\ee
which is equivalent to
\be
1=\sum_{y=1}^{d_B}\sqrt{2q_y}\mbf{v}_{y}^{\dag}\mbf{c}=2a\;.
\ee
We therefore conclude that
\be\label{vecc}
\mbf{v}_{y}^{\dag}\mbf{c}=\sqrt{\frac{q_y}{2}}\;.
\ee
Since $q_y>0$ we get that $\mbf{v}_y\neq 0$ for all $y$ and therefore $r_y>0$ for all $y$. Together with the orthogonality relation of $\mbf{v}_y$, this implies that the set of vectors $\left\{\frac{1}{\sqrt{r_y}}\mbf{v}_y\right\}$ is orthonormal. Therefore, the number of Kraus operators $m$ (which is the dimension of $\mbf{v}_{y|x}$) must be at least $d_B$. Hence, the equation above gives:
$$
\sum_{y=1}^{d_B}\frac{q_y}{2r_y}=\sum_{y=1}^{d_B}\frac{\mbf{c}^\dag\mbf{v}_{y}\mbf{v}_{y}^{\dag}\mbf{c}}{r_y}
\leq \mbf{c}^{\dag}\mbf{c}=1\;.
$$
A simple calculation shows that $\sum_y q_y/r_y$ obtains its minimum value when
\be\label{ry}
r_y=\frac{\sqrt{q_y}}{\sum_{y'=1}^{d_B}\sqrt{q_{y'}}}\;.
\ee
Therefore, we get,
\be
1\geq \sum_{y=1}^{d_B}\frac{q_y}{2r_y}\geq \frac{1}{2}\left(\sum_{y=1}^{d_B}\sqrt{q_y}\right)^2 
\ee
We therefore arrive at the following theorem.
\begin{theorem}
\label{thm:MIO-pure}
Let $|\psi\ra=\sqrt{p_0}|0\ra+\sqrt{p_1}|1\ra$ and $|\psi\ra=\sum_{y=1}^{d_B}\sqrt{q_y}|y\ra$, where
$q_y>0$ and $d_B>2$. Then, $|\psi\ra$ can be converted to $|\phi\ra$ if and only if $p_0=p_1=1/2$ and
\be\label{assumption}
\sum_{y=1}^{d_B}\sqrt{q_y}\leq\sqrt{2}\;.
\ee
\end{theorem}
\begin{proof}
The necessity of this condition follows from the arguments above.
To prove sufficiency, take $m=d_B+1$ and $\mbf{v}_y=\sqrt{r_y}\mbf{e}_y$, where $\{\mbf{e}_y\}$ is the standard basis of $\mathbb{C}^m$, and $r_y$ is given in~\eqref{ry}.
To be consistent with~\ref{vecc} we define for $j=1,...,d_B$
\be
c_j=\frac{\sqrt{q_j}}{\sqrt{2}\sum_{y=1}^{d_B}\sqrt{q_{y}}}
\ee
and for $j=d_B+1$ we define
\be
c_{d_B+1}=\sqrt{1-\sum_{j=1}^{d_B+1}c_{j}^{2}}\;.
\ee
Note that the term inside the sum is positive due to~\eqref{assumption}.
Finally, we define for $y=1,...,d_B$
\be
\mbf{u}_y=\sqrt{2q_y}\mbf{c}-\mbf{v}_y
\ee
With these choices, all the conditions in~\eqref{simp} are satisfied. This completes the proof.
\end{proof}
\begin{example}
Consider the following two states:
\be
|+\ra=\sqrt{\frac{1}{2}}|0\ra+\sqrt{\frac{1}{2}}|1\ra
\ee
and
\be
|\psi\ra:=\sqrt{\frac{8}{9}}|0\ra+\sqrt{\frac{1}{18}}|1\ra+\sqrt{\frac{1}{18}}|2\ra\;.
\ee
We show that the transformation $|+\ra\to|\psi\ra$ is achievable by maximally incoherent operations. Indeed, consider
the following three Kraus operators:
\begin{align}
& M_1=\frac{\sqrt{2}}{3\sqrt{3}}\begin{pmatrix}
3 & 1\\
0 & 1\\
0 & 1
\end{pmatrix}\\
& M_2=\frac{1}{3\sqrt{6}}\begin{pmatrix}
0 & 4\\
3 & -2\\
0 & 1
\end{pmatrix}\\
& 
M_2=\frac{1}{3\sqrt{6}}\begin{pmatrix}
0 & 4\\
0 & 1\\
3 & -2
\end{pmatrix}
\end{align}
\end{example}
It is strightforward to check that $\sum_{j=1}^{3}M_{j}^{\dag}M_j=I_2$ where $I_2$ is the $2\times 2$ identity matrix.
Furthermore, note that
\be
M_j|+\ra\propto 4|0\ra+|1\ra+|2\ra\propto |\psi\ra\;\;\forall\;j=1,2,3
\ee
To see that it is a maximal incoherent operation, note that
\be
\sum_{j=1}^{3}M_j|0\lr 0|M_{j}^{\dag}=\sum_{j=1}^{3}M_j|1\lr 1|M_{j}^{\dag}
=\frac{1}{6}\begin{pmatrix}
4 & 0 & 0\\
0 & 1 & 0\\
0 & 0 & 1
\end{pmatrix}\;.
\ee
\begin{figure}[t]
\includegraphics[scale=0.2]{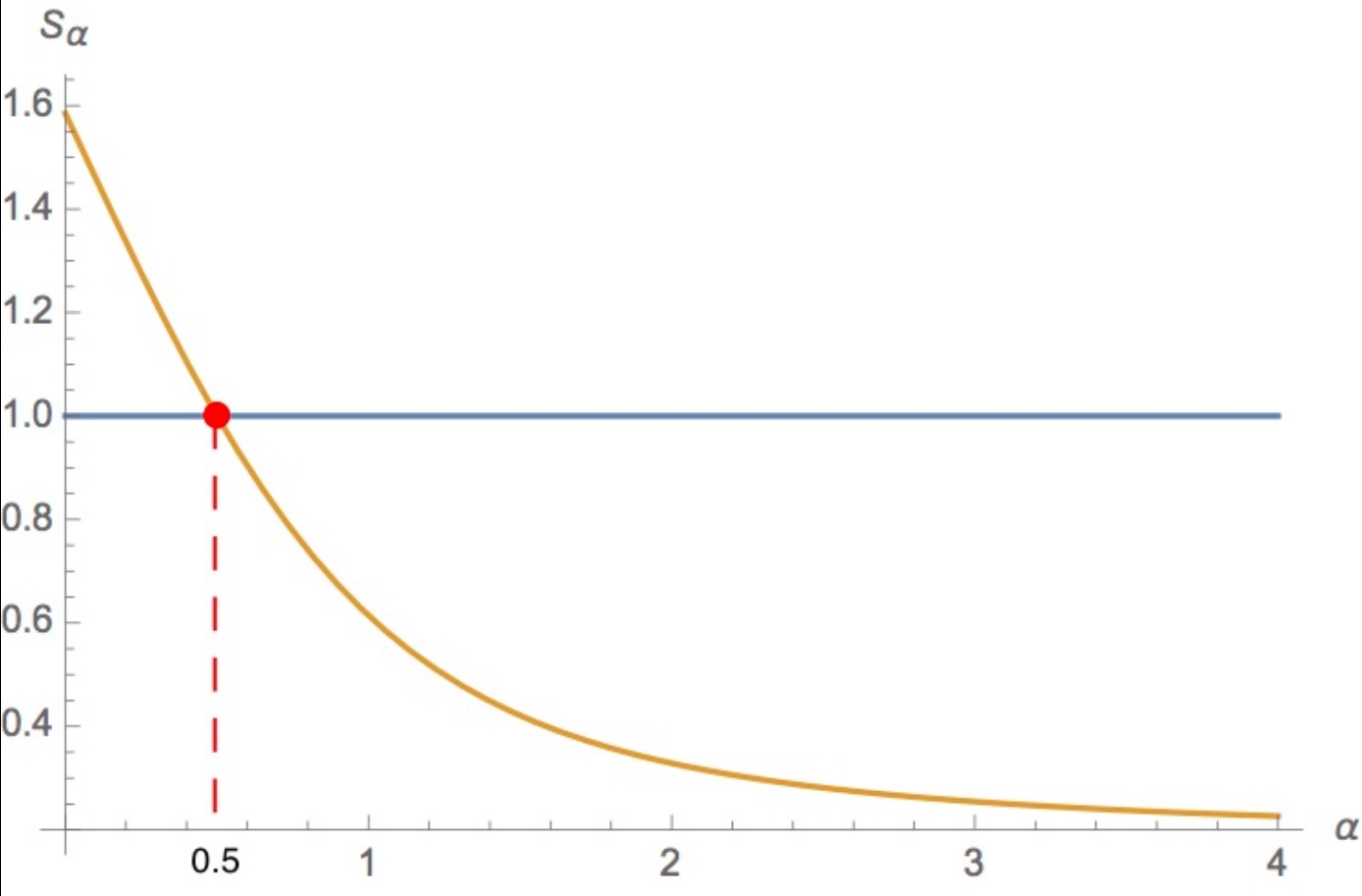}
\caption{\label{Fig1}  Comparison of $S_\alpha(|+\rangle)=1$ (the blue line) and $S_\alpha(|\psi\ra)$ (the yellow line) as a function of $\alpha$. For $0\leq \alpha<1/2$, $S_\alpha(|\psi\ra)> S_\alpha(|+\ra)$, and for $\alpha>1/2$, $S_\alpha(|\psi\ra)< S_\alpha(|+\ra)$.}
\end{figure}

In Fig.1 we plot the Renyi entropies of these two states. From the graph it is clear that 
$S_\alpha(|\psi\ra)> S_\alpha(|+\ra)=1$ for $\alpha\in[0,1/2)$. Therefore,
this example also demonstrate that all the Renyi entropies with $\alpha\in[0,1/2)$ are \emph{not} monotones and therefore are not measures of coherence. Furthermore, it provides an independant proof that the Renyi divergences $D_\alpha$ and $D_\alpha^{(q)}$ do \emph{not} satisfy the data processing inequality in the $\alpha$-ranges
$(2,\infty]$ and $[0,1/2)$, respectively.

\section{New family of monotones}
\begin{theorem}
Let $D(\rho\|\sigma)$ be a contractive function; i.e. $D(\mE(\rho)\|\mE(\sigma))\leq D(\rho\|\sigma)$
if $\mE$ is a CPTP map. Let $A_\rho$ be a set of density matrices acting on $\mbb{C}^d$. Note that the set $A_\rho$
can depend of the state $\rho$. If $\mE(A_\rho)\subseteq A_{\mE(\rho)}$ for all free operations $\mE$, then
then the two functions
\begin{align}
C_{A}^{R}(\rho)&= \min_{\sigma\in A_\rho}D(\rho\|\sigma)\nonumber\\
C_{A}^{L}(\rho) &=\min_{\sigma\in A_\rho}D(\sigma\|\rho)
\end{align}
are monotonic under the set of free operations.
\end{theorem}

\begin{proof}
\begin{align}
C_{A}^{R}(\mE(\rho)) 
& =\min_{\tau\in A_{\mE(\rho)}}D(\mE(\rho)\|\tau)\;\nonumber\\
&\leq\min_{\tau\in\mE( A_{\rho})}D(\mE(\rho)\|\tau)\;\nonumber\\
& =\min_{\sigma\in A_\rho} D(\mE(\rho)\|\mE(\sigma))\;\nonumber\\
& \leq\min_{\sigma\in A_\rho} D(\rho\|\sigma)=
C_{A}^{R}(\rho) 
\end{align}
Similar arguments prove that $C_{A}^{L}$ is also a monotone.
\end{proof}

\subsection{MIO Monotones}
\textbf{Example 1}: Take $A_\rho=\mI$ the set of incoherent diagonal states.
In this case, $A_\rho$ is independent of $\rho$ so we get trivially that
\be
\mE(A_\rho)=\mE(\mI)\subseteq\mI=A_{\mE(\rho)}\;
\ee
for any DIO (or MIO) $\mE$. Moreover, in this case,
\be
C_{A}^{R}(\rho)= \min_{\sigma\in \mI}D(\rho\|\sigma)
\ee
which is the well know measure we already discussed before. However, note that under PIO,SIO,IO,DIO or MIO
\be
C_{A}^{L}(\rho)= \min_{\sigma\in \mI}D(\sigma\|\rho)
\ee
is also a monotone.

\subsubsection{Relative R\'{e}nyi $\alpha$-monotones}

For $\alpha\in[0,\infty]$ the relative Renyi entropy is defined by
\be\label{type1alpha}
D_{\alpha}(\rho\|\sigma):=\frac{1}{\alpha-1}\log\tr(\rho^\alpha\sigma^{1-\alpha})\;.
\ee
This quantity is contractive (or equivalently satisfies the data processing inequality) for all $\alpha\in[0,2]$.
We will therefore be interested here only in this range of $\alpha$.
Define the $\alpha$-Coherence monotone by ($0\leq\alpha\leq 2$):
\be
C_\alpha(\rho):=\min_{\sigma\in\mI} D_{\alpha}(\rho\|\sigma)
\ee
We can compute this monotone explicitly, and part of the following work overlaps with independent work conducted by Rastegin in Ref. \cite{Rastegin-2015a}.  Let $\sigma=\sum_x q_x |x\lr x|$ be some free state.
Then,
\be
C_\alpha(\rho):=\min_{\{q_x\}} \frac{1}{\alpha-1}\log\sum_{x} q_{x}^{1-\alpha}\la x|\rho^\alpha|x\ra
\ee
Denote,
\be
r_{x}\equiv\frac{\left(\la x|\rho^\alpha|x\ra\right)^{1/\alpha}}{r}\;\;\;\text{where }\;\;r\equiv\sum_{x}\left(\la x|\rho^\alpha|x\ra\right)^{1/\alpha}
\ee
By definition, $\sum_x r_x=1$ and $r_x\geq 0$. Therefore,
\begin{align}
C_\alpha(\rho)&=\frac{\alpha}{\alpha-1}\log r+\min_{\{q_x\}}\frac{1}{\alpha-1}\log\sum_{x} q_{x}^{1-\alpha}r_{x}^{\alpha}\nonumber\\
&=\frac{\alpha}{\alpha-1}\log r+\min_{\{q_x\}}D_{\alpha}(\{r_x\}\|\{q_x\})\nonumber\\
& =\frac{\alpha}{\alpha-1}\log r\;,
\end{align}
where $D_{\alpha}(\{r_x\}\|\{q_x\})$ is the classical Renyi-divergence.
We therefore conclude that for $\alpha\in[0,2]$ the quantities
\be\label{alpha}
C_{\alpha}(\rho)=\frac{\alpha}{\alpha-1}\log \sum_{x}\left(\la x|\rho^\alpha|x\ra\right)^{1/\alpha}\;.
\ee
are coherence monotones.
Note that in the limit $\alpha\to 1$ we get $C_\alpha(\rho)\to C_{rel}(\rho)$.
Furthermore, in terms of the completely dephasing map $\Delta(\rho):=\sum_{x}\la x|\rho |x\ra\;|x\lr x|$, we have
\begin{align}\label{alpha}
C_\alpha(\rho)& =\frac{\alpha}{\alpha-1}\log\tr\left[\left(\Delta(\rho^\alpha)\right)^{1/\alpha}\right]\nonumber\\
&=
\frac{1}{\alpha-1}\log\tr\left[\left\|\Delta(\rho^\alpha)\right\|_{1/\alpha}\right].
\end{align}
$C_\alpha(\rho)$ can also be written in terms of the eigenvalues of $\rho$ as follows.
Suppose the spectrum decomposition of $\rho$ is given by
\be
\rho=\sum_{y=1}^{n}\lambda_y|v_y\lr v_y|
\ee
where $\lambda_y$ are the eigenvalues of $\rho$, with corresponding eigenvectors $|v_y\ra$.
Denote by $D$ the $n\times n$ doubly-stochastic matrix whose elements are $D_{xy}\equiv|\la x|v_y\ra|^2$. 
Then, Eq.~\eqref{alpha} takes the form
\be
C_{\alpha}(\rho)=\frac{\alpha}{\alpha-1}\log \sum_{x}\left(\sum_{y}D_{xy}\lambda_{y}^{\alpha}\right)^{1/\alpha}\;.
\ee

Note that for a pure state $\rho=|\psi\lr\psi|$ we have
\be
C_{\alpha}(\psi)=\frac{\alpha}{\alpha-1}\log\sum_jp_{j}^{1/\alpha}= S_{1/\alpha}(p)
\ee
where $S_{1/\alpha}$ is the R\'{e}nyi entropy with parameter $1/\alpha\in[1/2,\infty]$.
\begin{example}
Consider $\alpha=2$ in~\eqref{alpha}. Then, this monotone has a particular simple expression. Denoting by $\rho_{xy}$ the components of $\rho$ we get:
\be
C_{\alpha=2}(\rho)=2\log \sum_{x}\sqrt{\la x|\rho^2|x\ra}=2\log\sum_{x}\left(\sum_{y}|\rho_{xy}|^2\right)^{1/2}
\ee
We now apply this to the qubit case where
\be
\rho=\begin{pmatrix}
p & r \\
r & 1-p
\end{pmatrix}
\ee
Then,
\be
C_{\alpha=2}(\rho)=2\log\left(\sqrt{p^2+r^2}+\sqrt{(1-p)^2+r^2}\right)
\ee
\end{example}

\subsubsection{Quantum Relative R\'{e}nyi $\alpha$-monotones}

For $\alpha\in[1/2,\infty]$ the quantum relative Renyi entropy is given by
\be\label{type2}
D_{\alpha}^{(q)}(\rho\|\sigma):=\frac{1}{\alpha-1}\log\tr\left[\left(\sigma^{\frac{1-\alpha}{2\alpha}}\rho\sigma^{\frac{1-\alpha}{2\alpha}}\right)^\alpha\right]\;.
\ee
Define the quantum $\alpha$-Coherence monotone by:
\be
C_{\alpha}^{(q)}(\rho):=\min_{\sigma\in\mI} D^{(q)}_{\alpha}(\rho\|\sigma)
\ee
The minimization in this case is harder to perform.
However, for a pure state $\rho=|\psi\lr\psi|$ we have
\begin{align*}
D_{\alpha}^{(q)}(\rho\|\sigma)& =\frac{1}{\alpha-1}\log\tr\left[\left(\sigma^{\frac{1-\alpha}{2\alpha}}|\psi\lr\psi|\sigma^{\frac{1-\alpha}{2\alpha}}\right)^\alpha\right]\\
& = \frac{\alpha}{\alpha-1}\log\la\psi|\sigma^{\frac{1-\alpha}{\alpha}}|\psi\ra
\end{align*}
which is very similar to the expression we get for the relative Renyi entropy.
We therefore conclude that for pure states:
\be
C_{\alpha}^{(q)}(\psi)=\frac{2\alpha-1}{\alpha-1}\log\left(\sum_jp_{j}^{\frac{\alpha}{2\alpha-1}}\right)\;.
\ee
Denoting $\gamma\equiv\frac{\alpha}{2\alpha-1}$ we can rewrite the expression above as:
\be
C_{\alpha}^{(q)}(\psi)=\frac{1}{1-\gamma}\log\left(\sum_jp_{j}^{\gamma}\right)\equiv S_\gamma(\mbf{p})\;.
\ee
 
Note that the range of $\gamma$ is also $[1/2,\infty]$. Also, the other two parameter quantum divergences
introduced in~\cite{Datta15} lead to the same R\'{e}nyi entropies for pure states. Therefore, one may be tempted to conjecture that the transformation
\be
|\psi\rangle\to|\phi\rangle
\ee
is possible by MIO if and only if
\be
S_\alpha(\mbf{p})\geq S_\alpha(\mbf{q})\;\;\forall\;\alpha\in[1/2,\infty]\;,
\ee
where the probability vectors $\mbf{p}$ and $\mbf{q}$ corresponds to $|\psi\ra$ and $|\phi\ra$, respectively.
However, note that the requirements $p_0=p_1=\frac{1}{2}$ in Theorem \ref{thm:MIO-pure} 
shows that this conjecture is false. That is, the above equation is necessary 
but not sufficient for the existence of a MIO from $|\psi\ra\to|\phi\ra$.

\begin{example}
Consider the case $\alpha=\infty$ in~\eqref{type2}.
In this case, $D^{(q)}_\alpha$, is known to be equal to the max relative entropy given by
\be
\label{Eq:Renyi-Rel-Entropy-Inf1}
D^{(q)}_\infty(\rho\|\sigma)=\log\min\{\lambda\;:\;\rho\leq\lambda\sigma\}
\ee
The corresponding monotone is therefore
\be
C_\infty^{(q)}(\rho)=\log\min\left\{\tr(\sigma)\;:\;\rho\leq\sigma\;\;;\;\;\frac{\sigma}{\tr(\sigma)}\in\mI\right\}
\ee
To calculate this expression, observe that it can be rewritten as
\be
\label{Eq:Renyi-Rel-Entropy-Inf2}
C^{(q)}_\infty(\rho)=\log\min\left\{\tr(\sigma)\;:\;\rho\leq\Delta(\sigma)\;\;;\;\;\sigma\geq 0\right\}
\ee
Next, we recall the dual formulation in linear programming (see, e.g. Renes' paper on sub-relative-majorization \cite{Renes-2015a}, as well as recent work by Piani \textit{et al.} \cite{Piani-2016a}).
Consider the following setting of linear programming. Let $V_1$ and $V_2$ be two (inner product) vector spaces with two cones $K_1\subset V_1$ and $K_2\subset V_2$. Consider two vectors $v_1\in V_1$ and $v_2\in V_2$, and a linear map
$\mathcal{T}:V_1\to V_2$. Then, the primal form:
\be
\max_{\substack{x\in K_1\\
v_2-\mT(x)\in K_2}}\langle v_1,x\rangle_1
\ee
The dual form involve $\mT^*:V_2\to V_1$:
\be
\min_{\substack{y\in K_2\\
\mT^*(y)-v_1\in K_1}}\langle v_2,y\rangle_2
\ee
Applying this to our formulation, take $V_1=V_2=H_n$ the vector space of $n\times n$ Hermitian matrices.
Take $K_1=K_2=H_{n,+}$ be the cone of positive semi-definite matrices in $H_n$. Take $\mT=\Delta$ which is self-adjoint. Finally, take $v_2=I$, $v_1=\rho$, $y=\sigma$, $x=\tau$. With this choices the dual is our original expression for $C_\infty$
and the primal is the following expression
\begin{align}
\label{Eq:Cdual1}
C_\infty^{(q)}(\rho)&=\log\max\left\{\tr(\rho\tau)\;:\;\Delta(\tau)\leq I\;\;;\;\;\tau\geq 0\right\}\\
&=\log\max\left\{\tr(\rho\tau)\;:\;\Delta(\tau)= I\;\;;\;\;\tau\geq 0\right\}
\end{align}
Note that for $j\neq k$, $|\tau_{jk}|\leq 1$. Otherwise, if $|\tau_{jk}|>1$, one can find $\theta\in[0,2\pi]$ such that
for $|\psi\ra=|j\ra+e^{i\theta}|k\ra$, the expectation value $\la\psi|\tau|\psi\ra<0$. We therefore conclude that
\begin{align}
\tr(\rho\tau)=1+\sum_{j\neq k}\rho_{jk}\tau_{kj}&\leq 1+\sum_{j\neq k}|\rho_{jk}|\;\notag\\
&=1+C_{\ell_1}(\rho),
\end{align}
where 
\begin{equation}
C_{\ell_1}(\rho)=\sum_{j\neq k}|\rho_{jk}|
\end{equation}
is the so called $\ell_1$ coherence measure \cite{Baumgratz-2014a}.  This bound can be saturated in the case where $\rho$ is real with non-negative off-diagonal terms, in which case we take $\tau=|\psi\lr\psi|$ with $|\psi\ra=\sum_{x}|x\ra$. 

Note the relation between $C_\infty^{(q)}$ and the Robustness of Coherence $C_R$, which is defined as
\begin{align}\label{robust}
C_{R}(\rho)&=\min_{t\geq 0}\left\{t\;\Big|\;\frac{\rho+t\sigma}{1+t}\in\mI,\;\sigma\geq 0\right\}.
\end{align}
Letting $\hat{\sigma}=\rho+t\sigma$ so that $t=\tr[\hat{\sigma}]-1$, we can rewrite this as
\begin{align}
C_R(\rho)=\min_{\hat{\sigma}}\left\{\tr[\hat{\sigma}]-1\;\Big|\;\frac{\hat{\sigma}}{\tr[\hat{\sigma}]}\in\mI,\;\hat{\sigma}\geq \rho\right\}.
\end{align}
Putting everything together, we obtain
\begin{proposition}
\label{Prop:robustness}
\begin{equation}
C_\infty^{(q)}(\rho)=\log[1+C_R(\rho)].
\end{equation}
Moreover, $C_R(\rho)=C_{\ell_1}(\rho)$ for pure states, qubit mixed states, and any state $\rho$ with non-negative real matrix elements when expressed in the incoherent basis.
\end{proposition}
\end{example}
It is still an open problem whether $C_{\ell_1}$ is a MIO monotone in general, although it is a known monotone under IO \cite{Baumgratz-2014a}.

\subsection{DIO Monotones}

Take $A_\rho=\{\Delta(\rho)\}$ which contains only a single state.
Note that under DIO $\mE$ we have
\be
\mE(A_\rho)=\{\mE(\Delta(\rho))\}=\{\Delta(\mE(\rho))\}=A_{\mE(\rho)}\;.
\ee
Therefore, both the functions
\begin{align}
C_{A}^{R}(\rho)= D(\rho\|\Delta(\rho))\;\;\;\;,\;\;\;\;
C_{A}^{L}(\rho) =D(\Delta(\rho)\|\rho)
\end{align}
are monotones.
If we take $D(\rho,\sigma)=\|\rho-\sigma\|$, where $\|\cdot\|$ is the trace norm, we get
\be
C_{A}^{R}(\rho)=C_{A}^{L}(\rho)=\|\rho-\Delta(\rho)\|\;,
\ee
which is a function only of the off-diagonal terms. 

If we choose $D$ as in~\eqref{type1alpha} then we get the following monotones:
\begin{align}
C_{\alpha}^{R}(\rho)= \frac{1}{\alpha-1}\log\tr\left[\rho^\alpha\left(\Delta(\rho)\right)^{1-\alpha}\right]\nonumber\\
C_{\alpha}^{L}(\rho) =\frac{1}{\alpha-1}\log\tr\left[\left(\Delta(\rho)\right)^{\alpha}\rho^{1-\alpha}\right]
\end{align}
For a pure state $\rho=|\psi\lr\psi|$ with $|\psi\ra=\sum_{x}\sqrt{p_x}|x\ra$ we have
\begin{align}
C_{\alpha}^{R}(\rho)
&=\frac{1}{\alpha-1}\log\la\psi|\left(\Delta(\rho)\right)^{1-\alpha}|\psi\ra\nonumber\\
&=\frac{1}{\alpha-1}\log\sum_x p_{x}^{2-\alpha}\equiv\frac{1}{1-\gamma}\log\sum_x p_{x}^{\gamma}=S_\gamma(\rho)\;.
\end{align}
where we denoted $\gamma\equiv 2-\alpha$. Since the $C_\alpha^R$ is a DIO monotone for $\alpha\in[0,2]$, together with the fact that DIO$\subset$MIO, we have that all R\'{e}nyi entropies are DIO monotones. This is in contrast with the set of MIO for which $S_\gamma$ is a monotone only for $\gamma\geq 1/2$.

\subsubsection{$\Delta$-Robustness of Coherence}
Take 
\begin{align}
& A_\rho=\nonumber\\
& \left\{\frac{(1+t)\Delta(\rho)-\rho}{t}\;\Big |\;t>0\;;\;\;(1+t)\Delta(\rho)-\rho\geq 0\right\}
\end{align}
In this case, it is straightforward to check that $\mE(A_\rho)\subseteq A_{\mE(\rho)}$ for all $\mE\in$ DIO.  We consider the quantum R\'{e}nyi relative entropy $C_{\Delta,\alpha}^{(q)}(\rho):=\min_{\sigma\in A_\rho}D^{(q)}(\rho||\sigma)$.  Then in the limit $\alpha\to\infty$, we obtain analogs to Eqns. \eqref{Eq:Renyi-Rel-Entropy-Inf1} and \eqref{Eq:Renyi-Rel-Entropy-Inf2}:

\begin{align}
&C_{\Delta,\infty}^{(q)}(\rho)\notag\\
&=\log\min\left\{\tr(\sigma)\;\Big |\;\rho\leq\sigma\;\;;\;\;\frac{\sigma}{\tr(\sigma)}\in\ A_\rho\right\}\notag\\
&=\min_{t,\lambda>0}\left\{\lambda\;\Big |\;\rho\leq\lambda\frac{(1+t)\Delta(\rho)-\rho}{t}\;\;;\;\;(1+t)\Delta(\rho)\geq\rho\right\}\notag\\
&=\min_{t,\lambda>0}\left\{\lambda\;\Big |\;\frac{t+\lambda}{\lambda}\rho\leq (1+t)\Delta(\rho)\;\;;\;\;(1+t)\Delta(\rho)\geq\rho\right\}\nonumber\\
&=\min_{t,\lambda>0}\left\{\lambda\;\Big |\;\frac{t+\lambda}{\lambda}\rho\leq (1+t)\Delta(\rho)\right\}\;,
\end{align}
where the last equality follows from the fact that $\frac{t+\lambda}{\lambda}\geq 1$.  Note that $0\leq\lambda\frac{(1+t)\Delta(\rho)-\rho}{t}-\rho$, which means that $0\leq(\lambda-1)\Delta(\rho)$; thus, $\lambda\geq 1$.  Then the minimum above can be written as
\be
\min_{t,\lambda>0}\left\{\lambda\;:\;\frac{t+\lambda}{1+t}\rho\leq \lambda\Delta(\rho)\right\}
\ee
But since $\frac{t+\lambda}{1+t}>1$ we must have $\rho\leq \lambda\Delta(\rho)$. On the otherhand, taking the limit $t\to\infty$ in the above minimum gives $\rho\leq \lambda\Delta(\rho)$. We therefore conclude that the above minimum is equal to 
\be
\min_{\lambda>0}\left\{\lambda\;:\;\rho\leq \lambda\Delta(\rho)\right\}
\ee
or equivalently
\be
1+\min_{t>0}\left\{t\;:\;\rho\leq (1+t)\Delta(\rho)\right\}\;.
\ee
Finally, note that $t\geq 0$ satisfies $\rho\leq (1+t)\Delta(\rho)$ iff there exists a matrix $\sigma$ such that (i) $\frac{\rho+t\sigma}{1+t}\in\mI$, (ii) $\sigma\geq 0$, and (iii) $\Delta(\sigma)=\Delta(\rho)$.  Therefore we have the $\Delta$ analog of Prop. \ref{Prop:robustness}:
\begin{equation}
C_{\Delta,\infty}^{(q)}(\rho)=\log[1+C_{\Delta,R}(\rho)],
\end{equation}
where $C_{\Delta,R}(\rho)$ is a quantity we shall call the $\Delta$-Robustness of Coherence:
 \begin{align}
& C_{\Delta,R}(\rho):=\nonumber\\
& \min\left\{t\geq 0\;\Big|\;\frac{\rho+t\sigma}{1+t}\in\mI\;,\;\sigma\geq 0\;,\;\Delta(\sigma)=\Delta(\rho)\right\}.
 \end{align}
By construction, $C_{\Delta,R}$ is a DIO monotone.
\begin{example}
Consider the qubit state
\be
\rho=\begin{pmatrix}
p & r \\
r & 1-p
\end{pmatrix}
\ee
Then, the matrix $\sigma$ must have the form
\be
\sigma=\begin{pmatrix}
p & -\frac{r}{t} \\
-\frac{r}{t} & 1-p
\end{pmatrix}\;,
\ee
to ensure that $\rho+t\sigma$ is diagonal and $\Delta(\sigma)=\Delta(\rho)$. 
Now, the condition $\sigma\geq 0$ gives a lower bound on $t$. We therefore conclude
that for $0<p<1$
\be
C_R(\rho)=\frac{r}{\sqrt{p(1-p)}}
\ee
and otherwise, for $p=0$ or $p=1$, $C_R(\rho)=0$.
\end{example}

The form of $\sigma$ above can be generalized to any dimention.
That is, for $\rho=\Delta(\rho)+Z$, $\sigma$ must have the form 
\be
\sigma=\Delta(\rho)-\frac{1}{t}Z
\ee
Hence, $C_R(\rho)$ equals the minimum values of $t\geq 0$ such that $\sigma$ above is positive semidefinite.
Note that the positivity of $\sigma$ is equivalent to the positivity of
\be
t\Delta(\rho)-Z=t\Delta(\rho)-\left(\rho-\Delta(\rho)\right)=(1+t)\Delta(\rho)-\rho
\ee
We therefore arrive at the following expression for $C_R$:
\begin{align}
C_{\Delta,R}(\rho)& =\min\left\{t\geq 0\;\Big|\;(1+t)\Delta(\rho)-\rho\geq 0\right\}\nonumber\\
& =\max\left\{\frac{\la\phi|\rho|\phi\ra}{\la\phi|\Delta(\rho)|\phi\ra}\;\Big|\;|\phi\ra\in\mathbb{C}^d\;,\;\la\phi|\phi\ra=1\right\}
\end{align}

\begin{theorem}\label{cp}
Consider the linear map 
\be
\Phi_t(\rho)\equiv(1+t)\Delta(\rho)-\rho\;.
\ee
The following are equivalent:

(1) $\Phi_t(\rho)$ is positive

(2) $\Phi_t(\rho)$ is completely positive

(3) The parmeter $t\geq d-1$
\end{theorem}
\begin{proof}
The Choi matrix
\begin{align*}
&I\otimes\Phi_t(|\psi^+\lr\psi^+|)\\
&=\sum_{j,k}|j\lr k|\otimes \Phi_t(|j\lr k|)\\
& = \sum_{j}|j\lr j|\otimes \Phi_t(|j\lr j|)+\sum_{j\ne k}|j\lr k|\otimes \Phi_t(|j\lr k|)\\
& =t\sum_{j}|j\lr j|\otimes |j\lr j|-\sum_{j\ne k}|j\lr k|\otimes |j\lr k|\\
&=(1+t)\sum_{j}|j\lr j|\otimes |j\lr j|-|\psi^+\lr\psi^+|
\end{align*}
Finally, note that the last term is positive if and only if $1+t\geq d$. This complete the proof that $(2)$ and $(3)$ are equivalent. It is therefore left to show that $(1)$ implies $(3)$. To see it, note that
\be
\Phi_t(|+\lr+|)=\frac{1+t}{d}I-|+\lr+|
\ee
where $|+\ra\equiv\frac{1}{\sqrt{d}}\sum_{j}|j\ra$. Since we assume that $\Phi_t$ is positive, it follows that $1+t\geq d$.
\end{proof}
\begin{corollary}
The function
\be
R_D(\rho):=\log\left(1+C_R(\rho)\right)
\ee
which we call \emph{logarithmic robustness of dephasing} is a faithful measure of coherence (i.e. $R_D(\rho)=0$ iff $\Delta(\rho)=\rho$) satisfying 
\be
0\leq R_D(\rho)\leq\log d
\ee 
\end{corollary}

\begin{conjecture}
$R_D$ is additive. It is true for pure states (see below), unknown for mixed states.
\end{conjecture}

\begin{lemma}
For a pure state $|\psi\ra=\sum_{x=1}^{n}\sqrt{p_x}|x\ra$, with  $n\leq d$ and $p_x>0$,
\be
C_{R}(|\psi\ra)=n-1\;.
\ee
\end{lemma}

\begin{proof}
Let $|\phi\ra=\sum_{x=1}^{n}\sqrt{q_x}e^{i\theta_x}|x\ra$
then
\begin{align}
\frac{\la\psi|\rho|\psi\ra}{\la\psi|\Delta(\rho)|\psi\ra}& =\frac{\sum_{x\neq x'}\sqrt{p_xq_xp_{x'}q_{x'}}e^{i(\theta_x-\theta_{x'})}}{\sum_xp_xq_x}\nonumber\\
& \leq \frac{\sum_{x\neq x'}\sqrt{p_xq_xp_{x'}q_{x'}}}{\sum_xp_xq_x}\nonumber\\
&=\mbf{u}^{\dag}A\mbf{u}
\end{align}
where $\mbf{u}$ is a unit vector in $\mbf{C}^n$ with components
\be
u_x\equiv \frac{\sqrt{p_xq_x}}{\sqrt{\sum_{x'=1}^{n}p_{x'}q_{x'}}}
\ee
and $A$ is the $n\times n$ matrix
\be
A=\begin{pmatrix}
0 & 1 & 1 & \cdots & 1\\
1 & 0 & 1 & \cdots & 1\\
1 & 1 & 0 & \cdots & 1\\
\vdots & \vdots & \vdots & \ddots & \vdots\\
1 & 1 & 1 & \cdots & 0\\
\end{pmatrix}\;.
\ee
Hence, by taking $\theta_x=0$ and 
\be
q_x=\frac{1}{p_x}\Big/\sum_{x'=1}^{n}\frac{1}{p_{x'}}
\ee
we get that $\mbf{u}=\frac{1}{\sqrt{n}}(1,...,1)^T$ corresponds to the maximal eigenvalue
of $A$; i.e. for this choice $\mbf{u}^{\dag}A\mbf{u}=n-1$. This completes the proof.
\end{proof}

\section{Qubit Coherence}

In this section we focus exclusively on maps whose input/output space consists of single qubit density matrices.  We will say that a qubit state $\rho$ is in \textbf{standard form} when expressed as
\begin{equation}
\rho=\begin{pmatrix} p&r\\r&1-p\end{pmatrix}\quad p\geq 1/2,\;r\geq 0
\end{equation}
in the incoherent basis.  Any state $\rho$ can always be transformed into standard form by an incoherent unitary transformation, and thus each state can be uniquely parametrized by the tuple $(p,r)$ with $p\geq 1/2$, $r\geq 0$.

\subsection{Channels: IO-MIO Equivalence}

The main result we prove here is that every MIO channel $\mc{E}$ has a Kraus operator implementation that belongs to IO.  
\begin{theorem}
\label{Thm:IO=DIO}
IO=MIO for CPTP maps $\mc{E}:\mc{B}(\mbb{C}^2)\to\mc{B}(\mbb{C}^2)$.
\end{theorem}
\begin{proof}
Consider an arbitrary MIO CPTP map $\mc{E}$ with Kraus operator representation $\{M_j\}_{j=0}^t$.  We want to prove that $\mE$ has another Kraus operator representation with each operator having one of the forms given in Eq. \eqref{Eq:Kraus-IO}.  Since $\mc{E}$ is MIO, we have
\begin{equation}
\label{Eq:CPcons1}
\sum_{j=0}^{m-1}\bra{y}M_j\op{x}{x}M_j^\dagger\ket{y\oplus 1}=0\qquad \forall x,y\in\{0,1\}.
\end{equation}
Our goal is to find another Kraus operator representation $\{\tilde{M}_j\}_{j=0}^{\tilde{t}}$ of the channel $\mc{E}$ such that 
\begin{equation}
\label{Eq:KrausQubit}
\bra{y}\tilde{M}_j\op{x}{x}\tilde{M}_j^{\dagger}\ket{y\oplus 1}=0\qquad\forall x,y\in\{0,1\},\forall j.\end{equation}
We describe iteratively how this can always be done.  In the following recall that Kraus operators $\{\tilde{M}_j\}_{j=0}^{\tilde{t}}$ generate the same channel $\mc{E}$ iff $\tilde{M}_j=\sum_{k=0}^{m-1}u_{jk}M_k$ for some unitary matrix $u_{jk}$.
\begin{enumerate}
\item Take $x=0$.  Find two distinct values $(j,j')$ such that $\bra{0}M_j\op{x}{x}M_j^\dagger\ket{1}\not=0$ and $\bra{0}M_{j'}\op{x}{x}M_{j'}^\dagger\ket{1}\not=0$; relabel and denote these by $(j,j')=(0,1)$.  If two distinct values cannot be found, then by Eq. \eqref{Eq:CPcons1} we must have that $\bra{0}M_j\op{x}{x}M_j^\dagger\ket{1}=0$ for all $j$, and in which case set $\tilde{M}_j=M_j$ for all $j$ and proceed to step 4.  Otherwise, proceed to step 2.
\item  Consider an $m\times m$ unitary matrix whose only non-trivial action consists of a $2\times 2$ block $\left(\begin{smallmatrix}u_{00}&u_{01}\\u_{10}&u_{11}\end{smallmatrix}\right)$.  Then a different Kraus operator representation for $\mc{E}$ is realized by the elements $\tilde{M}_i=u_{i0}M_0+u_{i1}M_1$ for $i=0,1$ and $\tilde{M}_i=M_i$ for $i=2,\cdots, m-1$.  The unitary matrix is chosen such that $(u_{00},u_{01})$ is the normalized vector of $(-\bra{0}M_1\ket{x},\bra{0}M_0\ket{x})$.  With this choice, we have 
\begin{align}
\bra{0}\tilde{M}_0\ket{x}=u_{00}\bra{0}M_0\ket{x}+u_{01}\bra{0}M_1\ket{x}=0.
\end{align}
\item Repeat step 1. with the updated set of Kraus operators $\{\tilde{M}_0,\tilde{M}_1,\tilde{M}_i\}_{i=2}^{m-1}$.
\item At this step in the procedure, we have a Kraus representation $\{\tilde{M}_j\}_{j=0}^{m-1}$ for $\mc{E}$ such that either $\bra{0}\tilde{M}_j\ket{x}=0$ or $\bra{1}\tilde{M}_j\ket{x}=0$ for all $j$.
\item Repeat the previous steps except with choosing $x=1$.  In the end, we obtain an ensemble satisfying Eq. \eqref{Eq:KrausQubit}.  This completes the procedure.
\end{enumerate}
 \end{proof}

\subsection{Transformations: SIO-DIO-IO-MIO Equivalence}

We now proceed to show that in terms of a single incoherent transformation $\rho\to\sigma$, MIO is just as powerful as SIO. Since SIO is both a subset of IO and DIO it follows that SIO=OI=DIO=MIO on qubits.   As demonstrated above, the Robustness of Coherence and the $\Delta$-Robustness of Coherence for qubits can be computed explicitly:
\begin{align}
C_R(\rho)&=2r\notag\\
C_{\Delta,R}(\rho)&=\frac{r}{\sqrt{p(1-p)}}.
\end{align}
In general $C_R$ is a MIO monotone while $C_{\Delta, R}$ is DIO monotone.  However,  we will now show that $C_{\Delta,R}$ is also a MIO monotone for qubits.  

\begin{theorem}
\label{Thm:Robustness-MIO-Monotone}
$C_{\Delta,R}$ is monotonic under MIO channels $\mc{E}:\mc{B}(\mbb{C}^2)\to\mc{B}(\mbb{C}^2)$.
\end{theorem}
\begin{proof}
By Theorem \ref{Thm:IO=DIO}, it suffices to prove that $C_{\Delta,R}$ is an IO monotone.  For qubits, any CP map $\mc{E}$ that belongs to IO can always be expressed as
\begin{align}
\label{Eq:StandardCP}
\sigma=\mc{E}(\rho)&=\sum_\alpha J_\alpha \rho J_\alpha^\dagger +\sum_{\beta} K_\beta\rho K_\beta^\dagger \notag\\
&+\sum_{\gamma} L_\gamma\rho L_\gamma^\dagger +\sum_{\delta} M_\delta \rho M_\delta^\dagger, 
\end{align}
where the Kraus operators $\{J_\alpha,K_\beta,L_\gamma, M_\delta\}_{\alpha,\beta,\gamma,\delta}$ have the general form
\begin{align}
\label{Eq:Kraus-IO}
J_\alpha&=j_{\alpha 0}\op{0}{0}+j_{\alpha 1}\op{1}{1}\notag\\
K_\beta&=k_{\beta 0}\op{1}{0}+k_{\beta 1}\op{0}{1}\notag\\
L_\gamma&=l_{\gamma 0}\op{0}{0}+l_{\gamma 1}\op{0}{1}\notag\\
M_\delta&=m_{\delta 0}\op{1}{0}+m_{\delta 1}\op{1}{1}.
\end{align}
Crucially, these operators share the following relationships with $\Delta$:
\begin{align}
\Delta\left(J_\alpha\rho J_\alpha^\dagger\right)&=J_\alpha\Delta\left(\rho\right)J_\alpha^\dagger\notag\\
\Delta\left(K_\beta\rho K_\beta^\dagger\right)&=K_\beta\Delta\left(\rho\right)J_\beta^\dagger\notag\\
\Delta\left(L_\gamma\rho L_\gamma^\dagger\right)&=L_\gamma\rho L_\gamma^\dagger\notag\\
\Delta\left(M_\delta\rho M_\delta^\dagger\right)&=M_\delta\rho M_\delta^\dagger
\end{align}
for all $\rho$.  Suppose now that $t\geq 0$ satisfies $(1+t)\Delta(\rho)-\rho\geq 0$.  Then for an IO channel $\mc{E}$ we have
\begin{align}
(1+t)\Delta[\mc{E}(\rho)]-\mc{E}(\rho)=t \omega &+\sum_{\alpha}J_\alpha[(1+t)\Delta(\rho)-\rho]J_\alpha^\dagger\notag\\
&+\sum_{\beta}K_\beta[(1+t)\Delta(\rho)-\rho]K_\beta^\dagger,\notag
\end{align}
where
\[\omega=t \left(\sum_\gamma L_\gamma\rho L_\gamma^\dagger+\sum_\delta M_\delta\rho M_\delta^\dagger\right)\geq 0.\]
By the assumption $(1+t)\Delta(\rho)-\rho\geq 0$ we likewise have $(1+t)\Delta[\mc{E}(\rho)]-\mc{E}(\rho)\geq 0$.  From the definition of $C_{\Delta,R}$, it therefore follows that
\begin{equation}
C_{\Delta,R}(\rho)\geq C_{\Delta,R}(\mc{E}(\rho)).
\end{equation}
\end{proof}

Next, we prove that monotonicity of $C_{\Delta,R}(\rho)$ is also sufficient for an SIO (and therefore also MIO) transformation.
\begin{lemma}
\label{Lem:DIO-Transformation}
Let $\rho$ and $\sigma$ have standard-form parametrizations $(p,r)$ and $(q,t)$ respectively.  Then $\rho$ can be transformed into $\sigma$ by SIO if and only if
\begin{equation}
\label{Eq:qubit-robustness-conditions}
C_R(\rho)\geq C_R(\sigma)\qquad\text{and}\qquad C_{\Delta,R}(\rho)\geq C_{\Delta,R}(\sigma).
\end{equation}
\end{lemma}
\begin{proof}
We will describe a channel $\mc{E}$ consisting exclusively of Kraus operators having the form $J_\alpha$ and $K_\beta$ as given in Eq. \eqref{Eq:KrausQubit}.  The transformation will consist of two steps $\rho\to\sigma_{max}\to\sigma$, where $\sigma_{max}$ has parameters $(q,t_{max}(q))$ with
\be
t_{\max}(q)=\begin{cases} r\quad\text{if $p\geq q$}\\ r\sqrt{\frac{q(1-q)}{p(1-p)}}\quad\text{if $q\geq p$}.
\end{cases}
\ee
The channel attaining $t_{\max}$ is given by $\rho\mapsto \sigma_{max}=J\rho J^\dagger+ K\rho K^\dagger$, where
\begin{align}
j_0^2&=\begin{cases} \frac{p+q-1}{2p-1}\quad\text{if $p\geq q$}\\
\frac{q}{p}\frac{p+q-1}{2q-1}\quad\text{if $q\geq p$}\end{cases}\notag\\
j_1^2&=\begin{cases} \frac{p-q}{2p-1}\quad\text{if $p\geq q$}\\
\frac{1-q}{1-p}\frac{p+q-1}{2q-1}\quad\text{if $q\geq p$}\end{cases}\notag\\
k_0^2&=1-j_0^2\notag\\
k_1^2&=1-j_1^2.
\end{align}
Finally, the transformation $\sigma_{max}\to\sigma$ can be seen as SIO feasible by noting that any $t<t_{\max}(q)$ can be reached for a fixed value of $q$ by applying a dephasing channel $\rho=J_1\rho J^\dagger_1+J_2\rho J_2^\dagger$ where $J_1=\left(\begin{smallmatrix}\cos\theta&0\\0&\sin\theta\end{smallmatrix}\right)$ and $J_2=\left(\begin{smallmatrix}\sin\theta&0\\0&\cos\theta\end{smallmatrix}\right)$, for some appropriately chosen $\theta$.
\end{proof}
Combining Theorem \ref{Thm:Robustness-MIO-Monotone} with Lemma \ref{Lem:DIO-Transformation}, we therefore obtain the main result:
\begin{theorem}
\label{Thm:Qubit-Coherence}
For qubit states $\rho$ and $\sigma$, the transformation $\rho\to\sigma$ is possible by either DIO, IO, or MIO if and only if both $C_R(\rho)\geq C_R(\sigma)$ and $C_{\Delta,R}(\rho)\geq C_{\Delta, R}(\sigma)$.
\end{theorem}

\subsection{Coherence Measures}

For qubit states, a number of coherence measures have been proposed and evaluated, in direct analogy to entanglement measures in two-qubit systems.  For instance,    the so-called coherence of formation and concurrence of coherence \cite{Yuan-2015a, Winter-2015a} have been proposed, and both can be shown as being equivalent to the $\ell_1$-norm: $C_{\ell_1}(\rho)=2r$ \cite{Toloui-2011a, Yuan-2015a}.  Distinct from these is the relative entropy of coherence, which was known before under the name $G$-Asymmetry (see~\cite{GMS2009} and references therein), which takes the form
\begin{align}
C_{rel}(\rho)=S(\Delta(\rho))-S(\rho).
\end{align}
All measures in qubit systems can be seen as arising from the two robustness measures $C_{R}$ and $C_{\Delta,R}$ according to
\begin{align}
C_{\ell_1}(\rho)&=C_R(\rho)\notag\\
C_{rel}(\rho)&=f\left(\frac{C_R(\rho)}{C_{\Delta,R}(\rho)}\right)-f\left(\frac{C_R(\rho)}{C_{\Delta,R}(\rho)}\sqrt{1-C_{\Delta,R}(\rho)^2}\right),
\end{align}
where $f(x)=h\left(\frac{1}{2}[1-\sqrt{1-x^2}]\right)$ and $h(x)=-x\log x-(1-x)\log(1-x)$.

\section{Coherence Theories based on Asymmetry.}

\subsection{Translation Invariant Operations (TIO)}

\label{TIO}

Let us now comment further on asymmetry-based resource theories of coherence.  In these approaches, coherence is defined with respect to invariant subspaces of an observable $H$, say the Hamiltonian.  Specifically, one considers the unitary group of translations $\{e^{-it H}:t \in\mbb{R}\}$, and a state $\rho$ is said to be incoherent if it commutes with every element of the group; i.e. $e^{-itH}\rho e^{itH}=\rho$ for all $t$.  The class of translation invariant operations (TIO) consists of all CPTP maps $\mc{E}$ that commute with the unitary action of the group; i.e. 
$$\mc{E}[e^{-itH}(\rho) e^{itH}]=e^{-itH}[\mc{E}(\rho )]e^{itH}$$
 for all $t$ and all $\rho$.  The class TIO was first introduced and studied in Ref. \cite{Marvian-2015a}.  When $H$ is proportional to the number operator $\hat{N}$, then the unitary group of translations provides a representation for $U(1)$ \cite{Gour-2008a}.

TIO resource theory represents a specific example of an asymmetry-resource theory.  For a general compact group $G'$, a $G'$-asymmetry resource theory identifies its free states as those that are invariant under $G'$-twirling 
$$\mathcal{G}(\rho)=\int_{G'}dg U(g)\rho U(g)^{\dag},$$
 where $U: G'\to \mathcal{H}$ is the representation of $G'$ on the Hilbert space $\mathcal{H}$ and $dg$ the Haar measure.  The free operations are $G'$-covariant: 
 $$\mc{E}[U(g)\rho U(g)^\dagger]=U(t)[\mc{E}(\rho)]U(g)^\dagger$$
  for all $g\in G'$ and all $\rho$.  

Recall that under a basis-dependent definition of coherence, a state is incoherent if and only if it is diagonal in some specified basis $\mc{I}$, called the incoherent basis.  In order that a $G'$-asymmetry theory likewise identifies $\mc{I}$ as the free states, we need that $G'$ and its representation $U$ are such that 
$$\mathcal{G}(\rho)=\Delta(\rho).$$ 
The $G'$-twirling of $G'=N$ or $G'=U(1)$ (with representation $U(\theta)=e^{i\theta\hat{N}}$) both lead to $\mathcal{G}(\rho)=\Delta(\rho)$, provided that the representation uniquely decomposes into a direct product of one-dimensional subrepresentations.  Below we will show that the group $N$ is the largest group with this property, whereas its subgroup $U(1)$ is one of the smallest ones. 

In the case of TIO, the condition that $\mc{G}(\rho)\in\mc{I}$ amounts to the generator $H$ having a non-degenerate spectrum.  But in general, degeneracies will exist and the resulting resource theory will look very different than the basis-dependent theories of PIO/SIO/IO/DIO/MIO.  

As an example illustrating the sharp distinction between TIO and PIO/SIO/IO/DIO/MIO, consider a pair of bosons such as the electrons of a helium atom.  Due to the exchange symmetry, a natural incoherent basis to consider for this system is $\{\ket{b_0}=\sqrt{1/2}(\ket{01}+\ket{10}),\;\ket{b_1}=\sqrt{1/2}(\ket{01}-\ket{10}),\; \ket{b_2}=\ket{00},\;\ket{b_3}=\ket{11}\}$.  In the basis-dependent theories of PIO/SIO/IO/DIO/MIO, a state of this system is incoherent if and only if it is diagonal in this basis.  However, in a coherence resource theory based on $U(1)$-asymmetry, $\ket{b_0}$ and $\ket{b_1}$ are still identified as incoherent states, but so is the superposition state $\ket{\psi}=\sqrt{1/2}(\ket{b_0}+\ket{b_1})$ as well as the mixture $\rho=1/2(\op{b_0}{b_0}+\op{b_1}{b_1})$.  Typically $\ket{\psi}$ is called a coherent superposition whereas $\rho$ is an incoherent superposition.  

A TIO resource theory of coherence can thus be interpreted as defining coherence with respect to just individual degrees of freedom for a system, whereas a basis-dependent definition of coherence considers \textit{all} degrees of freedom.  In this sense, a basis-dependent theory of coherence may be seen as capturing a more complete notion of coherence for a system.  In terms of the generator $H$, TIO theory characterizes coherence between different \textit{eigenspaces} of $H$ rather than among a specific set of \textit{eigenstates}.  In certain settings in may be desirable to think of coherence in this way~\cite{Marvian-2015a}.

\subsection{$G$-Asymmetry and $N$-Asymmetry Resource Theories}

The set of all incoherent unitary matrices forms a group which we denote by $G$. The group $G$ consists of all $d\times d$ unitaries of the form $\pi u$, where $\pi$ is a permutation matrix and $u$ is a diagonal unitary matrix (with phases on the diagonal). We denote by $N$ the group of $d\times d$ diagonal unitary matrices and by $\Pi$ the group of permutation matrices. Note that $N$ is a normal subgroup of $G$, and $G=N\rtimes \Pi$ is the semi-direct product of $N$ and $\Pi$. Clearly, the group $G$ is compact and the twirlings over $N$ and $G$ are given by: 
\be
\int_N dg\;\mT_g( \rho)=\Delta(\rho)\;\text{ and }\;\int_G dg\;\mT_g( \rho) =\frac{1}{d}I\;
\ee
where $\mT_g(\rho):=g\rho g^\dag$, and the integration is with respect to the Haar measure $dg$.

\subsubsection{$G$-covariant maps}

We would like to characterize the set of all $G$-covariant quantum channels. 
That is, we would like to characterize all CPTP maps that satisfies
\be
[\mE,\mT_g]=0\;\;,\;\;\forall\;g\in G\;.
\ee
Consider the following 3 CPTP maps that are all G-covariant: 
\begin{align}
& \mE^{(1)}(\rho)=\rho\nonumber\\
& \mE^{(2)}(\rho)=\frac{1}{d-1}\left(I-\Delta(\rho)\right)\nonumber\\
& \mE^{(3)}(\rho)=\frac{1}{d-1}\left(d\Delta(\rho)-\rho\right)
\end{align}
\begin{remark}
\textbf{(1)} The map $\mE^{(1)}$ is the trivial map and it is covariant under \emph{all} groups (with unitary representations),
whereas the last two maps are non-trivial as they are not covariant with respect to all groups.\\
\textbf{(2)} The two convex combinations of $\mE^{(1)}$, $\mE^{(2)}$, and $\mE^{(3)}$:
\begin{align*}
& \frac{1}{d^2}\mE^{(1)}(\rho)+\frac{d-1}{d}\mE^{(2)}(\rho)+\frac{d-1}{d^2}\mE^{(3)}(\rho)=\frac{1}{d}I
\nonumber\\
&\frac{d}{d+1}\mE^{(2)}(\rho)+\frac{1}{d+1}\mE^{(3)}(\rho)=\frac{1}{d^2-1}(dI-\rho)
\end{align*}
are also covariant under all groups (note that the coefficient $d$ in front of $I$ in RHS of the second equation is necessary since otherwise the map is not completely positive).\\ \textbf{(3)} The map $\mE^{(3)}$ is completely positive (see Theorem~\ref{cp}) and the coefficient $d$ in front of $\Delta(\rho)$ is necessary since otherwise the map is not positive.\\
\textbf{(4)} The dephasing map is the following convex combination of $\mE^{(1)}$ and $\mE^{(3)}$:
\be
\Delta(\rho)=\frac{1}{d}\mE^{(1)}(\rho)+\frac{d-1}{d}\mE^{(3)}(\rho)
\ee
\end{remark} 
 
The following theorem shows that up to convex combinations, these 3 CPTP maps are all the G-covariant maps.
\begin{theorem}$\;$\label{gcom}

{\bf{\rm (a)}} Let $G$ be as above, $U$ be a unitary matrix, and  $\mU(\rho):=U\rho U^\dag$.
Then, 
\be
[\mU,\Delta]=0\;\;\iff\;\;U\in G\;.
\ee

{\bf{\rm (b)}} A CPTP map $\mE$ is G-covariant if and only if
$\mE$ is a convex combination of the three CPTP maps defined above.
Explicitly, $\mE$ is G-covariant if and only if
\begin{align}
\mE(\rho)&=q_1\rho+
\frac{q_2}{d-1}\left(I-\Delta(\rho)\right)+
\frac{q_3}{d-1}\left(d\Delta(\rho)-\rho\right)
\end{align}
for some $q_i\geq 0$ with $\sum_{i=1}^{3}q_i=1$.
\end{theorem}

\begin{proof} (a) A direct calculation shows that $\Delta$ is a G-covariant map (it also follows from part B).
Conversely, suppose $[\Delta,\mU]=0$. Note that that for a given fixed $x$
\begin{align}
& \Delta\left(\mU(|x\lr x|)\right)=\sum_{x'}|\la x'|U|x\ra|^2|x'\lr x'|\nonumber\\
& \mU\left(\Delta(|x\lr x|)\right)=U|x\lr x|U^\dag
\end{align}
Comparing the two expressions gives $\la x'|U|x\ra=0$ except for one values of $x'$. Hence, $U\in G$.
\end{proof}

Before, we prove part (b) of the theorem, we first prove the following lemma:
\begin{lemma}
\label{ncov}
Let $\mE$ be an N-covariant CPTP map; that is,
\be
[\mE,\mT_g]=0\;\;,\;\;\forall\;g\in N\;.
\ee
Then, $\mE$ has the following Kraus decomposition
\be\label{basic1}
\mE(\rho)=\sum_j M_j\rho M_{j}^{\dag}+\sum_{x\neq x'}J_{xx'}\rho J_{xx'}^{\dag}
\ee
where all $M_j=\sum_xa_{jx}|x\lr x|$ are diagonal matrices and $J_{xx'}=b_{xx'}|x\lr x'|$.
\end{lemma}

\begin{proof}
We will apply Lemma 1 of~\cite{Gour-2008a} to the characterization of
$N$-invariant operations. Note first that the irreducible
representations of $N\cong U(1)^d$ are labeled by d integers 
$\mbf{k}=(k_1,...,k_d)$, and are all 1-dimensional. The
$\mbf{k}^{\text{th}}$ irreducible representation
$
u_{\mbf{k}}:N\to \mathbb{C}
$
has the form
\be
u_{\mbf{k}}(\vec{\theta})=e^{i\vec{\theta}\cdot\mbf{k}}\;.
\ee
where $\vec{\theta}=(\theta_1,...,\theta_d)\in U(1)^d$.
It follows from Lemma~1 of~\cite{Gour-2008a} 
that the Kraus operators
$K_{\mbf{k},\alpha}$
of a $N$-invariant operation can be labeled by the irrep
$\mbf{k}$ and a multiplicity index $\alpha$,
and satisfy
\be\label{formn}
g_{\vec{\theta}}\;K_{\mbf{k},\alpha}\;g_{\vec{\theta}}^{\dag}=e^{i\vec{\theta}\cdot\mbf{k}}K_{\mbf{k},\alpha}\;\;,\;\forall\vec{\theta}\in U(1)^d\;.
\ee
where $g_{\vec{\theta}}$ is the diagonal matrix with components $e^{i\theta_1},...,e^{i\theta_d}$ on the diagonal. 

Note that by virtue of the fact that the irreps are 1d,
the Kraus operators do not get mixed with one another
under the action of $N$  (this provides a significant
simplification relative to non-Abelian groups).
The most general expression for $K_{\mbf{k},\alpha}$ is
\be
K_{\mbf{k},\alpha}=\sum_{x,x'}c^{\mbf{k},\alpha}_{xx'}|x\lr x'|\;,
\ee
with some coefficients $c^{\mbf{k},\alpha}_{xx'}$.
Plugging this into~\eqref{formn} yields the constraint
\be
c^{\mbf{k},\alpha}_{xx'}\left(e^{i(\theta_x-\theta_{x'})}-e^{i\vec{\theta}\cdot\mbf{k}}\right)=0\;\;,\;\forall\;\vec{\theta}\in U(1)^d
\ee
Hence, $c^{\mbf{k},\alpha}_{xx'}$ must be zero unless $\mbf{k}=0$  and $x=x'$, 
or the $x$ and $x'$ components of 
$\mbf{k}$ are $1$ and $-1$, respectively, and all other components are zero. This completes the proof of the lemma. 
\end{proof}

Note that the lemma above provide the form of the Kraus operators in the resource theory of symmetric operations under the group $N$. This can be viewed as a physical resource theory of coherence. However, as discussed in the paper, resource theories of asymmetry cannot be used for coherence due to decoherence subspaces.
Moreover, as we can see from the above form of the Kraus operators, in the resource theory of $N$-asymmetry
permutations are not free!
We now ready to prove theorem~\ref{gcom}

\begin{proof} 
In addition to the form in~\eqref{basic1}, $\mE$ also has to commute with all permutations:
\be
[\mE,\mT_\pi]=0\;,\;\;\forall\;\pi\in\Pi\;.
\ee
In particular, we get
\begin{align}
&\mT_{\pi}\left(\mE(\rho)\right)=\nonumber\\
&\sum_{j,x,x'}a_{jx}\bar{a}_{jx'}\rho_{xx'}|\pi(x)\lr \pi(x')|\notag\\
&+\sum_{x'\neq x}|b_{xx'}|^2\rho_{x'x'}|\pi(x)\lr \pi(x)|
\end{align}
whereas
\begin{align}
\mE\left(\mT_{\pi}(\rho)\right)&=
\sum_{j,x,x'}a_{j\pi(x)}\bar{a}_{j\pi(x')}\rho_{xx'}|\pi(x)\lr \pi(x')|\nonumber\\
&+\sum_{x'\neq x}|b_{\pi(x)\pi(x')}|^2\rho_{x'x'}|\pi(x)\lr \pi(x)|
\end{align}
Hence, comparing the off-diagonal terms of $\mE\left(\mT_{\pi}(\rho)\right)=\mT_{\pi}\left(\mE(\rho)\right)$  give
\be
\sum_ja_{jx}\bar{a}_{jx'}=\sum_ja_{j\pi(x)}\bar{a}_{j\pi(x')}\equiv c\;,
\ee
since $\mE\left(\mT_{\pi}(\rho)\right)=\mT_{\pi}\left(\mE(\rho)\right)$ holds for all $\rho$ and for all permutations $\pi\in\Pi$. The constant $c\in\mathbb{R}$ and is independent of $x$ and $x'$
Comparing the diagonal terms of $\mE\left(\mT_{\pi}(\rho)\right)=\mT_{\pi}\left(\mE(\rho)\right)$ gives
\begin{align}
& \sum_j|a_{jx}|^2\rho_{xx}+\sum_{\{x':x'\neq x\}}|b_{xx'}|^2\rho_{x'x'}\nonumber\\
&=\sum_j|a_{j\pi(x)}|^2\rho_{xx}+\sum_{\{x':x'\neq x\}}|b_{\pi(x)\pi(x')}|^2\rho_{x'x'}\;\;\forall\;\rho
\end{align}
Since the equation above holds for all $\rho$ we must have
\be
\sum_j|a_{jx}|^2=\sum_j|a_{j\pi(x)}|^2\equiv a
\ee
and
\be
|b_{xx'}|^2=|b_{\pi(x)\pi(x')}|^2\equiv b\;,
\ee
where $a$ and $b$ are non-negative real numbers independent of $x$ and $x'$. We therefore get that
\begin{align}
&\mE(\rho)=\nonumber\\
&\sum_{x}a\rho_{xx}|x\lr x|+\sum_{x\neq x'}c\rho_{xx'}|x\lr x'|
+\sum_{x'\neq x}b\rho_{xx}|x'\lr x'|\nonumber\\
& =a\Delta(\rho)+c\left(\rho-\Delta(\rho)\right)+b\sum_x\rho_{xx} (I-|x\lr x|)\nonumber\\
&=a\Delta(\rho)+c\left(\rho-\Delta(\rho)\right)+b\left(I-\Delta(\rho)\right)
\end{align}
Note that the condition $\sum_jM_{j}^{\dag}M_j+\sum_{x\neq x'}J_{xx'}^{\dag}J_{xx'}=I$ gives
\be
a+b(d-1)=1\;.
\ee
We therefore conclude
\be\label{finally}
\mE(\rho)=a\Delta(\rho)+c\left(\rho-\Delta(\rho)\right)+\frac{1-a}{d-1}\left(I-\Delta(\rho)\right)
\ee
where $0\leq a\leq 1$. We now argue that 
\be\label{congg}
-\frac{a}{d-1}\leq c\leq a\;.
\ee 
Indeed,
\be
|c|\leq \sum_j|a_{jx}\bar{a}_{jx'}|\leq \sum_j\frac{1}{2}\left(|a_{jx}|^{2}+|a_{jx'}|^2\right)=a
\ee
and we also have
\begin{align}
0&\leq \sum_j\left(\sum_xa_{jx}\right)\left(\sum_{x'}\bar{a}_{jx'}\right)=\nonumber\\
&\sum_{x}\sum_j|a_{jx}|^2+\sum_{x\neq x'}\sum_ja_{jx}\bar{a}_{jx'}=
da+d(d-1)c\;,\nonumber
\end{align}
which is equivalent to $c\geq -a/(d-1)$. Finally, we note that~\eqref{finally} can be expressed as:
\begin{align}
\mE(\rho)& =\frac{a+c(d-1)}{d}\mE^{(1)}(\rho)+(1-a)\mE^{(2)}(\rho)\nonumber\\
&+\frac{(a-c)(d-1)}{d}\mE^{(3)}(\rho)
\end{align}
The constraints on $c$ in~\eqref{congg} ensures that the above equation is a convex combination of 
$\mE^{(1)}$, $\mE^{(2)}$, and $\mE^{(3)}$. This completes the proof of the theorem.
\end{proof}

\subsubsection{$N$-covariant maps}
The $N$-covariant operations given in Lemma~\ref{ncov} are very similar to the "cooling operations" given in~\cite{Varun}. The only difference is that $J_{xx'}$ is zero unless $x<x'$ (in the context of thermodynamics, the $x$ index corresponds to energy levels, and cooling operations can not increase the energy). Therefore, $N$-covariant operations are a bit more powerful than cooling operations, as can be seen from the following theorem, when compared with Theorem~1 in~\cite{Varun}.

\begin{theorem}
Let $\rho,\sigma$ be two density matrices of the same dimensions, with all the off-diagonal terms of $\rho$ being non-zero. Define the matrix $Q=(q_{xx'})$ as follows:
\be
q_{xx'}:=\left\{\begin{array}{ll}
& \min\left\{\frac{\sigma_{xx}}{\rho_{xx}},\;1\right\}\;\;\;\mbox{if  $x=x'$}\\
& \frac{\sigma_{xx'}}{\rho_{xx'}}\;\;\;\;\;\;\;\;\;\;\;\;\;\;\;\;\;\;\mbox{if  $x\neq x'$}
\end{array}\right.
\ee
Then, $\sigma=\mE(\rho)$ where $\mE$ is $N$-invariant operation if and only if $Q\geq 0$.
\end{theorem}
\begin{proof}
Let $\mbf{a}_x\equiv(a_{jx})_j$ where $a_{jx}$ are the coefficients of $M_j$ as in Eq.~\eqref{basic1}. Denote also
$h_{xx'}\equiv \mbf{a}_{x}^{\dag}\mbf{a}_{x'}$ , and
\be
r_{x'|x}\equiv\left\{\begin{array}{ll}
h_{xx} &\text{ if }\;\;  x=x'\\
|b_{xx'}|^2 &\text{ if }\;\; x\neq x'
\end{array}\right.
\ee
where $b_{xx'}$ are the coefficients associated with the operator $J_{xx'}$ in Eq.~\eqref{basic1}.
Since $\mE$ is trace preserving, $\sum_{x'}r_{x'|x}=1$. Note that the matrix $H=(h_{xx'})$ is Gramian and therefore positive semi-definite. Recall also that the components of any positive semi-definite matrix can be written as $\mbf{a}_{x}^{\dag}\mbf{a}_{x'}$ for some vectors $\mbf{a}_x$. Hence,
from~\eqref{basic1}
it follows that there exists $N$-covariant map $\mE$ such that $\sigma=\mE(\rho)$ iff there exists $H\geq 0$ and a column stochastic matrix $R=(r_{x|x'})$ with diagonal elements $r_{x|x}=h_{xx}$ such that
\be\label{rhosigma}
\sigma_{xx'}\equiv\left\{\begin{array}{ll}
\sum_{y}r_{x|y}\rho_{yy} &\text{ if }\;\;  x=x'\\
h_{xx'}\rho_{xx'} &\text{ if }\;\; x\neq x'
\end{array}\right.
\ee
From the relation above we get
\begin{align}
h_{xx'} & =\frac{\sigma_{xx'}}{\rho_{xx'}}\equiv q_{xx'}\;\;\;\text{for }\;\;x\neq x'\nonumber\\
h_{xx} &=r_{x|x}\leq\min\left\{\frac{\sigma_{xx}}{\rho_{xx}},1\right\}\equiv q_{xx}
\end{align}
Suppose now that $\sigma=\mE(\rho)$. Then, there exists $H\geq 0$ that satisfies the above relations. Since, $Q$ and $H$ are only different in the diagonal elements we can write $Q=H+D$ where $D$ is some diagonal matrix. The equation above shows that $D\geq 0$. Therefore, $Q\geq 0$. Conversely, suppose $Q\geq 0$. We need to show that there exists $H\geq 0$ and column stochastic matrix $R$ (with the same diagonal as $H$) that satisfy Eq.(\ref{rhosigma}). 
We take $H=Q$ and show that there exists $R$ with the desired properties. For simplicity of the exposition here, suppose that $\rho_{xx}\leq \sigma_{xx}$ for $x=1,...,k$ and $\rho_{xx}>\sigma_{xx}$ for $x=k+1,...,d$. We take the column stochastic matrix $R$ to have the following form
\be
R=\begin{pmatrix}
I_{k} & \;CD'\\
\mbf{0} & D
\end{pmatrix}
\ee
where $I_k$ is the $k\times k$ identity matrix, $\mbf{0}$ is the $(d-k)\times k$ zero matrix, $D$ is the $(d-k)\times(d-k)$ diagonal matrix with diagonal elements $\{\sigma_{xx}/\rho_{xx}\}$ with $x=k+1,...,d$, the matrix $C$ is a $k\times (d-k)$ column stochastic matrix, and $D'$ is a $(d-k)\times(d-k)$ diagonal matrix with diagonal elements $\{1-\sigma_{xx}/\rho_{xx}\}$ with $x=k+1,...,d$.   Hence, $R$ is column stochastic as long as $C$ is column stochastic.
With this form of $R$, 
the condition $\sigma_{xx}=\sum_y r_{x|y}\rho_{yy}$ is equivalent to
\be
\begin{pmatrix}
\sigma_{11}\\
\sigma_{22}\\
\vdots\\
\sigma_{kk}
\end{pmatrix}
=\begin{pmatrix}
\rho_{11}\\
\rho_{22}\\
\vdots\\
\rho_{kk}
\end{pmatrix}
+C\begin{pmatrix}
\rho_{(k+1)(k+1)}-\sigma_{k+1)(k+1)}\\
\rho_{(k+2)(k+2)}-\sigma_{(k+2)(k+2)}\\
\vdots\\
\rho_{dd}-\sigma_{dd}
\end{pmatrix}
\ee
Define $\mbf{r}$ to be the $k$-dimensional vector whose components are $\sigma_{xx}-\rho_{xx}$ for $x=1,...,k$, and 
$\mbf{t}$ the $d-k$-dimensional vector whose components are $\rho_{xx}-\sigma_{xx}$ for $x=k+1,...,d$. By definition, both vectors have non-negative components, and note also that the sum of the components of $\mbf{r}$ is the same as the sum of the components of $\mbf{t}$. Hence, there exists a column stochastic matrix $C$ that satisfies $\mbf{r}=C\mbf{t}$.
This completes the proof.
\end{proof}

In the next proposition we show that the group $N$ is the largest group possible with the property that its twirling is the dephasing map $\Delta$. 

\begin{proposition}
Let $G'$ be any group with unitary representation $U(g)$ for $g\in G'$ such that
\be
\int_{G'} dg\;U(g)\rho U(g)^{\dag}=\Delta\rho.
\ee
Then, the set $\{U(g)\}_{g\in G'}$ is a subgroup of $N$.
\end{proposition}
\begin{proof}
If $\int_{G'} dg\;U(g)\rho U(g)^{\dag}=\Delta\rho$, then
\be
\int_{G'}dg\;U(g)|x\ra\la x| U(g)^{\dag}=|x\ra\la x|\;,\;\;\;\forall\;x=1,...,d
\ee
which gives
\be
U(g)|x\ra=e^{i\theta_x(g)}|x\ra\;,
\ee
where $\{\theta_x\}_{x=1}^{d}$ are one-dimensional representations of $G'$.  
The equation above clearly indicates that $U(g)\in N$ so that $U(G')$ must be a subgroup of $N$.
In this sense, $N$ is the largest group with the property that $\mathcal{G}(\rho)=\Delta(\rho)$.

The requirement $\mathcal{G}(|x\lr x'|)=0$ for $x\neq x'$ gives in addition
\be
\int_{G'}dg e^{i\left(\theta_x(g)-\theta_{x'}(g)\right)}=\delta_{xx'}
\ee
Taking $dg=\frac{d\alpha}{2\pi}$ and $\theta_x(g)=x\alpha$ with $\alpha\in[0,2\pi]$ reproduce the $U(1)$-twirling.
Of course, the equation above is also satisfied for $\theta_x(g)=x^2\alpha$, but still the group $G'=U(1)$. \end{proof}

\section{Open problems}\label{final}

We conclude the Appendix with a few open questions.

\subsection{State Transformations}

Pure state transformations under SIO (both asymptotic and single copy cases) have been completely characterized in this paper via the one-to-one correspondence with LOCC. Consequently, among all coherence models discussed here, the SIO model is the most similar to the theory of pure bipartite entanglement. 
Particularly, in the single-copy regime, pure state transformations are determined by the majorization criterion (similar to Nielsen theorem in entanglement theory). A key open question is whether or not this criterion can be extended to the IO and DIO models.

Since majorization is both a necessary and sufficient condition for an SIO pure state transformation $|\psi\rangle\to\ket{\phi}$, it follows that it is sufficient for both IO and DIO (recall SIO is a subset of both IO and DIO). In IO it is also known to be necessary if both pure states have a full Schmidt rank since here the transformation is actually accomplished by sIO.  But as we discussed in this paper, it is not clear if it is still the case when the Schmidt rank of the target state $|\phi\ra$ is strictly smaller than the Schmidt rank of $|\psi\rangle$.

As for DIO, we have shown that all the R\'{e}nyi entropies of the Schmidt components of a pure state are monotones under DIO. In~\cite{Turgut} it was shown that if $S_\alpha(\psi)\geq S_{\alpha}(\phi)$ for all $\alpha$ then there exists a catalyst $|C\ra$ such that the Schmidt components of $|\psi\ra|C\ra$ are majorized by the Schmidt components of $|\phi\ra|C\ra$. Therefore, the existence of a catalyst provides a sufficient condition
for the transformation $|\psi\rangle\to\ket{\phi}$ under DIO. This means that necessary \emph{and} sufficient condition for pure state transformation under DIO are somewhere between majorization and catalytic majorization.

Majorization also provides sufficient condition for $|\psi\ra\to|\phi\ra$ under MIO, but here we also know that is is not necessary. In fact, MIO can increase the Schmidt rank as demonstrated in Theorem~\ref{thm:MIO-pure}.  However, Theorem~\ref{thm:MIO-pure} only involves a transformation from pure qubit to pure qudit.  It is left open to extend it to higher dimensions.

Necessary and sufficient conditions for mixed-state transformations have only been found for the qubit case, and a special type of asymmetry-based theory with symmetry groups $G$ and $N$. However, in higher dimensions, necessary and sufficient conditions for mixed state transformations for SIO/IO/DIO/MIO are not known. In the asymptotic limit of many copies of a mixed state we know that IO is not a reversible model, and distillation and formation rates have been calculated in~\cite{Winter-2015a}. MIO on the other hand, is a reversible quantum resource theory (QRT) in the asymptotic limit of many copies, due to a general QRT theorem proved in~\cite{Brandao2015}. However, the asymptotic distillation and formation rates are not known for SIO and DIO.   

Finally, another area of open inquiry pertains to determining the precise relationship between SIO, IO, and DIO.  To our knowledge, no operational gap in terms of state transformation is known between these classes, despite the fact that they represent distinct collections of CP maps.  More precisely, for every transformation $\rho\to\sigma$ feasible by IO (resp. DIO), is it also feasible by DIO (resp. IO) as well as SIO?  We suspect that such examples can be found, but perhaps not when $\rho$ is pure.

\subsection{Monotones}

There are few open problems regrading coherence monotones.
In~\cite{Baumgratz-2014a} a measure of coherence under IO was introduced.
This measure was defined by
\be
C_{\ell_1}(\rho)=\sum_{x\neq y}\rho_{xy}\;,
\ee
where $\rho_{xy}$ are components of $\rho$ in the incoherent basis. In the Appendix, we have shown that 
the robustness of coherence as defined in~\eqref{robust} equals $C_{\ell_1}$ for pure states and mixed states with
non-negative real off-diagonal terms. While the robustness of coherence is a monotone under MIO, it is not know if $C_{\ell_1}$ is also a monotone under MIO.

In the Appendix we have also introduced many new monotones under DIO. These set of monotones are closely related to monotones under thermal operations. In the resource theory of quantum theormodynamics, the free (or ``thermal'') operations take the form $\rho_A \to\tr_B[U(\rho_A\otimes \gamma_B^{(T)})U^\dagger]$, where $U$ is any unitary that commutes with the joint Hamiltonian, and $\gamma^{(T)}_B$ is the Gibbs state at temperature $T$ \cite{Janzing-2000a, Brandao-2013a}. It was also observed in~\cite{Terry} that Thermal operations are time-translation symmetric, and in particular belongs to DIO when the incoherent basis is taken to be the energy eigenstates, assuming no-degeneracy in the energy eigenstates. Therefore, all the DIO monotones introduced in this appendix, are also monotones under thermal operations.  In the case of degeneracy in the energy eigenstates, it is left open how to apply the DIO monotones to thermodynamics. 

\subsection{Relating Coherence with Maximally Correlated Entanglement}

Propositions \ref{Prop:SIO-LOCC} and \ref{Prop:sIO-LOCC} show that every transformation $\rho\to\sigma$ by either SIO or sIO corresponds to an LOCC transformation between the corresponding maximally correlated states $\rho^{(mc)}\to\sigma^{(mc)}$.  One obtains the maximally correlated state $\rho^{(mc)}$ from the single-system state $\rho$ via the ``coherent channel'' $\ket{x}\to\ket{xx}$.  In and of itself, such a channel appears in the theory of coherent communication where the tasks of coherent superdense coding and coherent teleportation are fully dual to one another (see Chapter 7 of \cite{Wilde-2013a}).  We have been interested in using this channel to map the theory of SIO/sIO into one-way/two-way LOCC.  A natural question is whether or not such a connection can also be established between IO and LOCC.  Such a relationship has been conjectured in Ref. \cite{Winter-2015a}, and a probabilistic version of it was proven in Ref. \cite{Chitambar-2016a}.  Specifically, it was shown that for every IO transformation $\rho\to\sigma$, the transformation $\rho^{(mc)}\to\sigma^{(mc)}$ can always be accomplished with some nonzero probability.  It is unknown whether a deterministic LOCC implementation is always possible, and whether such a result also holds for transformations $\rho\to\sigma$ that are feasible using DIO.

Lastly, Theorem \ref{thm:MIO-pure} shows that $\rho\to\sigma$ by MIO fails to imply $\rho^{(mc)}\to\sigma^{(mc)}$ by LOCC.  Unlike LOCC, MIO is able to increase the Schmidt rank under pure state transformations.  An interesting open question is whether, analogous to MIO, the Schmidt rank can be increased by some non-entangling operation.  

\bibliography{CoherenceBib}

\end{document}